\documentclass[12pt]{article}
\pdfoutput=1
\usepackage{amsmath,amssymb,amsthm,appendix,bm,graphicx,float,array,multirow,multicol,rotfloat,caption,subcaption,hyperref,cleveref,enumerate,geometry,mathdots,adjustbox,booktabs,parskip,mathtools,tikz,tikz-cd,pdflscape,csquotes,lscape,rotating,empheq,braket,ragged2e}
\usepackage[para]{threeparttable}
\usepackage[all]{xy}
\usepackage[normalem]{ulem}
\usepackage[numbers,sort&compress]{natbib}
\numberwithin{equation}{section}
\numberwithin{figure}{section}
\allowdisplaybreaks
\makeatletter
\usetikzlibrary{arrows, positioning, decorations.pathmorphing, decorations.markings, decorations.pathreplacing, decorations.markings, matrix, patterns}
\hypersetup{colorlinks=true}
\hypersetup{linkcolor=black}
\hypersetup{citecolor=black}\hypersetup{urlcolor=black}
\makeatletter

\theoremstyle{plain}
\newtheorem*{thm*}{Theorem}
\newtheorem{thm}{Theorem}[section]

\newtheorem{prop}[thm]{Proposition}

\theoremstyle{definition}
\newtheorem{defn}[thm]{Definition}
\newtheorem*{defn*}{Definition}

\crefname{lemma}{lemma}{lemmas}
\Crefname{lemma}{Lemma}{Lemmas}
\crefname{thm}{theorem}{theorems}
\Crefname{thm}{Theorem}{Theorems}
\crefname{defn}{definition}{definitions}
\Crefname{defn}{Definition}{Definitions}

\DeclarePairedDelimiterX{\abs}[1]{\lvert}{\rvert}{\ifblank{#1}{{}\cdot{}}{#1}}
\makeatother

\newtheorem*{thm:main}{Theorem \ref{thm:main}}
\newtheorem*{thm:prop}{Proposition \ref{thm:prop}}

\newcommand{\be}{\begin{equation}}
\newcommand{\ee}{\end{equation}}
\newcommand{\bea}{\begin{eqnarray}}
\newcommand{\eea}{\end{eqnarray}}
\newcommand{\bega}{\begin{gather}}
\newcommand{\eega}{\end{gather}}

\newcommand{\bi}{\begin{itemize}}
\newcommand{\ei}{\end{itemize}}
\newcommand{\ben}{\begin{enumerate}}
\newcommand{\een}{\end{enumerate}}
\newcommand{\bca}{\begin{cases}}
\newcommand{\eca}{\end{cases}}
\newcommand{\bln}{\begin{align}}
\newcommand{\eln}{\end{align}}
\newcommand{\bst}{\begin{split}}
\newcommand{\est}{\end{split}}
\def\ie{\begin{equation}\begin{aligned}}
\def\fe{\end{aligned}\end{equation}}
\newcommand{\bma}{\le(\begin{matrix}}
\newcommand{\ema}{\end{matrix}\ri)}
\newcommand{\bwt}{\begin{widetext}}
\newcommand{\ewt}{\end{widetext}}

\def\b{{\beta}}
\def\p{{\partial}}
\newcommand\ep{\epsilon}

\newcommand\lam{\lambda}

\newcommand\om{\omega}
\newcommand\Om{\Omega}

\newcommand\ga{{\ensuremath{{\gamma}}}}
\newcommand\Ga{{\ensuremath{{\Gamma}}}}
\newcommand\de{{\ensuremath{{\delta}}}}
\newcommand\De{{\ensuremath{{\Delta}}}}
\newcommand\vp{\varphi}

\newcommand\da{{\dagger}}

\def\th{{\theta}}

\newcommand\ov{\over}
\newcommand\ha{{\half}}

\newcommand\apr{{\ensuremath{{\alpha'}}}}
\def\le{\left}
\def\ri{\right}

\newcommand\sB{{\ensuremath{{\mathcal B}}}}

\newcommand\sH{{\ensuremath{{\mathcal H}}}}

\newcommand\sM{{\ensuremath{{\mathcal M}}}}
\newcommand\sN{{\ensuremath{{\mathcal N}}}}

\newcommand\sV{{\mathcal V}}

\newcommand\sS{{\mathcal S}}
\newcommand\sT{{\mathcal T}}
\newcommand\sW{{\mathcal W}}
\newcommand\sX{{\mathcal X}}

\newcommand\sZ{{\mathcal Z}}

\newcommand{\fa}{{\mathfrak{a}}}
\newcommand{\fb}{{\mathfrak{b}}}

\newcommand{\fo}{{\mathfrak o}}

\newcommand\vev[1]{{\ensuremath{\left\langle{#1}\right\rangle}}}

\newcommand{\field}[1]{\ensuremath{\mathbb{#1}}}

\newcommand{\RR}{\field{R}}

\newcommand\half{{\ensuremath{\frac{1}{2}}}}
\newcommand\vx{{\vec x}}
\newcommand\vk{{\vec k}}

\begin{document}
\begin{titlepage}
\vspace*{-3cm} 
\begin{flushright}
{\tt MIT-CTP/5751}\\
\end{flushright}
\begin{center}
\vspace{2.2cm}
{\LARGE\bfseries Toward stringy horizons}\\
\vspace{1cm}
{\large
Elliott Gesteau$^{1}$ and Hong Liu$^{2}$\\}
\vspace{.6cm}
{ $^1$ Division of Physics, Mathematics, and Astronomy, California Institute of Technology}\par\vspace{-.3cm}
{Pasadena, CA 91125, U.S.A.}\par
\vspace{.2cm}
{ $^2$ Center for Theoretical Physics, Massachusetts Institute of Technology}\par\vspace{-.3cm}
{Cambridge, MA 02139, U.S.A.}\par
\vspace{.4cm}
\scalebox{.95}{\tt  egesteau@caltech.edu, hong$\_$liu@mit.edu}\par
\vspace{1cm}
{\bf{Abstract}}\\
\end{center}
We take a first step towards developing a new language to describe causal structure, event horizons, and quantum extremal surfaces (QES) for the bulk description of holographic systems  beyond the standard Einstein gravity regime.
By considering the structure of boundary operator algebras, we introduce a stringy ``causal depth parameter'', which quantifies the depth of the emergent radial direction in the bulk, and a certain notion of ergodicity on the boundary. We define stringy event horizons in terms of the half-sided inclusion property, which is related to a stronger notion of boundary ergodic or quantum chaotic behavior.
Using our definition, we argue that above the Hawking--Page temperature, there is an emergent sharp horizon structure in the large $N$ limit of $\mathcal{N}=4$ Super-Yang--Mills at finite nonzero 't Hooft coupling. In contrast, some previously considered toy models of black hole information loss do not have a stringy horizon.
Our methods can also be used to probe violations of the equivalence principle for the bulk gravitational system, and to explore aspects of stringy nonlocality.

\vfill 

\end{titlepage}

\tableofcontents
\newpage
\section{Introduction}

In the semiclassical gravity description,\footnote{By semiclassical gravity, we mean classical Einstein gravity (with possible higher derivative corrections) and quantum field theory in a curved spacetime.} an event horizon is a boundary beyond which events cannot affect an asymptotic observer. While originally defined in terms of the causal structure of a spacetime, event horizons have also played crucial roles in many other aspects of black hole physics, including the discovery of Hawking radiation and black hole thermodynamics.

In string theory, the semiclassical gravity description serves as a low-energy effective theory, arising in the limit where Newton's constant $G_N$ and the string length $\ell_s$ (or $\apr =\ell_s^2$) both approach zero.
At finite $G_N$, spacetime fluctuations prevent sharp definitions of geometric concepts such as event horizons, causal structure, and spacetime subregions.  It is natural to ask whether these concepts can still be precisely formulated 
in the so-called ``stringy regime,'' where $G_N$ approaches zero while $\alpha'$ remains finite. 
In this regime, while there are no quantum gravitational fluctuations, fundamental objects are one-dimensional strings, which  probe spacetime very differently from point particles. In particular, sharp locality used in the current formulations of these concepts may not be compatible with the intrinsic non-locality brought by the finite size of these strings. 
Despite much progress in our understanding of string theory, there have been very limited tools for approaching this question.

In this paper, we propose a definition of event horizons for holographic gravitational systems in the stringy regime in terms of their boundary duals. As a prototypical example, consider $\sN=4$ Super-Yang-Mills (SYM) theory with gauge group $SU(N)$, which is dual to IIB string theory on AdS$_5 \times S_5$. The semi-classical regime of the bulk gravity theory corresponds to the $N \to \infty$ and $\lam \to \infty$ ($\lam$ is the 't Hooft coupling)  limit  of the SYM theory, while the stringy regime corresponds to $N \to \infty$ with $\lam$ kept finite. Suppose the SYM theory is in the thermofield double state 
above the Hawking-Page temperature~\cite{Hawking:1982dh}, which in the $N \to \infty$ and $\lam \to \infty$ limit is dual to an eternal black hole in AdS$_5$~\cite{Wit98b,Mal01}. Now consider the system at finite $\lam$. 
We will argue that the bulk is described by a stringy black hole and provide tools to explore the corresponding stringy horizon.

The starting point of our proposal is the subregion-subalgebra duality introduced in~\cite{LeuLiu21b,LeuLiu22}, 
which provides a new framework for understanding the emergence of bulk spacetime from the boundary theory
in the semi-classical regime. 
The duality  states that a bulk subregion $\fo$ can be described by a subalgebra $\sM_\fo$ of the boundary theory. 

$\sM_\fo$ includes all physical operations that can be performed by bulk observers in the region $\fo$, and encodes its geometric, causal, and entanglement structures. 
For example, the inclusion of one bulk spacetime region $\fo_1$ in another region $\fo_2$ is described by inclusion of the corresponding dual boundary subalgebras $\sM_{\fo_1} \subset \sM_{\fo_2}$. That $\fo_1$ and $\fo_2$ are causally disconnected corresponds to the fact that $\sM_{\fo_1}$ and $\sM_{\fo_2}$ commute. The entanglement structure of a region $\fo$ is reflected in that $\sM_\fo$ is a type III$_1$ von Neumann algebra and associated algebraic properties. 

Philosophically, subregion-subalgebra duality is reminiscent of Gelfand duality in mathematics, where a topological space can be equivalently described by the algebras of functions on that space. Here the collection of boundary subalgebras $\{\sM_\fo\}$ captures not only the geometric information on the bulk spacetime, but also all the physics defined on it.

Now suppose the boundary subalgebra $\sM_\fo$ for a bulk subregion $\fo$ can be extended to finite $\lam$. 
We can consider using $\sM_\fo$ as a {\it definition} of the bulk region $\fo$ in the stringy regime. 
If that can be shown to be sensible, the commutation relations among different algebras can then be used to define the causal structure of the emergent bulk spacetime in the stringy regime, which should in turn enable the definition of an event horizon.

\begin{figure}
\centering
\includegraphics[scale=0.5]{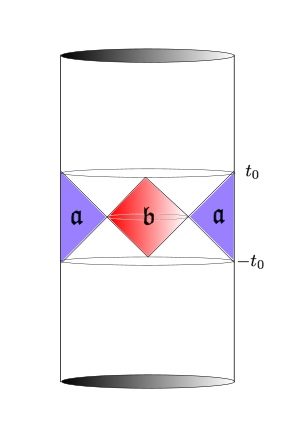}
\caption{A simple illustration of subregion/subalgebra duality in the vacuum state of strongly coupled $\mathcal{N}=4$ Super-Yang--Mills theory. The region $\fa$, a spherical Rindler region in the bulk, is dual to the algebra $\mathcal{S}_I$ of large $N$ boundary observables in the time band $I=(-t_0,t_0)$. The bulk diamond $\fb$, which does not touch the boundary, can be identified with the commutant of the algebra $\mathcal{S}_I$.}
\label{fig:vtime}
\end{figure}

As an illustration, consider the algebra $\sS_{I}$ generated by single-trace operators localized  within a time band $I = (-t_0, t_0)$ in the vacuum sector of the boundary theory\footnote{In the $N \to \infty$ limit, boundary operators at different times are independent, and the space of states in the boundary theory splits into disjoint sectors associated with semi-classical states.}. In the large $\lam$ limit, it has been argued~\cite{LeuLiu21b,LeuLiu22} to be dual to a spherical Rindler region $\fa$ in empty AdS (see Fig.~\ref{fig:vtime}). The bulk diamond region $\fb$, which is the causal complement of $\fa$ in the bulk, can be described by the boundary subalgebra $\sS_{I}'$, the commutant of $\sS_{I}$. 
It is then natural to use $\sS_{I}$ and $\sS_{I}'$ at finite $\lam$ to define the corresponding subregions $\fa$ and $\fb$ in the stringy regime. 

In the large $N$ limit, the bulk theory is free, and the boundary theory is a generalized free field theory. The algebra $\sM_\fo$ dual to a bulk region $\fo$ is then recovered from all the algebras generated by individual single-trace operators. This feature extends to finite $\lam$. In particular, if we consider a specific spacetime field corresponding to a stringy excitation, it should still be described by a free quantum field theory in a curved spacetime. Similarly, on the boundary each single-trace operator is still described by a generalized free field, and the discussion of operator algebras associated with each single-trace operator goes exactly as that in the $\lam \to \infty$ limit. 

There are, however,  two possible complications in using $\sM_\fo$ as a definition of $\fo$ in the stringy regime: 

\ben 

\item In the semiclassical limit, all bulk fields in the Einstein gravity ``see'' the same geometry, which is the statement of the ``equivalence principle''. In the stringy regime, the equivalence principle may be broken, indications of which can already be seen perturbatively in $\apr$ by including higher derivative corrections. Due to such corrections, different bulk fields may ``see'' different effective spacetime metrics\footnote{
An explicit example was given in~\cite{BriLiu07}, where adding the Gauss-Bonnet term to the Einstein action leads to different effective metrics for different components of metric perturbations.}.

Thus it can be that in the stringy regime, sharp definitions of subregions or horizons can be given for each bulk free field, but the definitions differ for different fields. That is, a universal definition that is common to all fields may not exist. In particular, interactions among the fields may lead to a nonlocal smearing of the boundaries or horizons of subregions, see Fig.~\ref{fig:blurry} for a cartoon.

\begin{figure}
\centering
\includegraphics[scale=0.5]{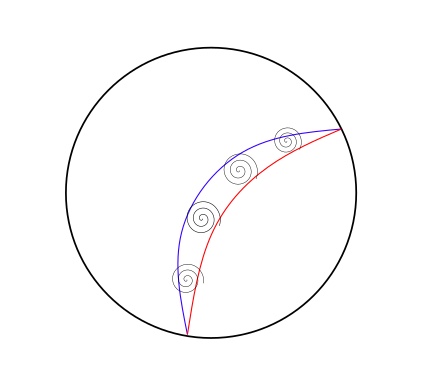}
\caption{A cartoon of how the notion of spacetime may break down in the stringy regime. The bulk region ``seen" by a boundary field or another from a given boundary subregion may differ depending on the field - two such regions are represented in blue and red. If these two fields interact, the difference between these two spacetime geometries may become blurry.}
\label{fig:blurry}
\end{figure}

It may also happen in certain stringy spacetimes, sharp subregions in general cannot be defined, but an event horizon universal to all stringy fields may still be sharply defined.

In terms of  boundary subalgebras, in the $\lam \to \infty$ limit, the equivalence principle implies that properties of different factors in $\sM_\fo$ associated with different single-trace operators should ``conspire'' to describe the same bulk region $\fo$. 
Non-existence of a definition universal to all fields should then be reflected in the breaking of this ``conspiracy''  among factors of $\sM_\fo$ by finite $\lam$ effects.

\item There is a Hagedorn growth in the number of single-trace operators in $\sM_\fo$, dual to stringy modes in the bulk. 
Given that the Hagedorn growth is what distinguishes a string theory from  a field theory with an infinite number of fields\footnote{IIB supergravity on AdS$_5 \times S_5$ already has an infinite number of fields in AdS$_5$ from the Kaluza-Klein reduction on $S_5$.}, it may be expected to have fundamental implications for the structure of $\sM_\fo$.

\een

Despite the above possible complications, $\sM_\fo$ provides a framework and powerful new tools where questions about 
bulk stringy geometry can be explored. For example, identifying sharp signatures that algebras of different single-trace operators ``reconstruct'' different bulk spacetime geometries may give insight into the nature of stringy nonlocality, and lead to clues on how to characterize stringy geometry. 

Motivated by the above considerations,  we will use the structure of boundary algebras to explore geometric aspects of the bulk in the stringy regime, and propose a diagnostic for existence of stringy horizons. Here is a summary of the main results of the paper:

\ben 

\item A causal depth parameter, which is also well-defined in the stringy regime, is introduced to quantify the ``depth'' of the emergent radial direction. 

In the large $N$ limit, the causal depth parameter can be calculated for each single-trace operator by using its boundary spectral function.  It measures the depth of the bulk as probed by the bulk field dual to the single-trace operator. Different values of the causal depth parameter for different single-trace operators would be a direct signature that in the stringy regime different fields ``see'' different bulk geometries.  

 We show that for a CFT in the vacuum state (dual to empty AdS), all bulk fields have the same causal depth parameter. The same statement applies to the thermofield double state below the Hawking-Page temperature. 

The computation of the causal depth parameter exactly maps to a well-studied problem from harmonic analysis, which is known as the \textit{exponential type problem}. In many physically relevant situations, this connection makes it possible to gain information about the causal depth parameter, in particular whether it is finite or not, without knowing the specific form of the spectral function.

\item We give a boundary diagnostic for the existence of stringy horizons 
using the causal depth parameter and the existence of half-sided modular inclusions. 

  As an example, we argue that the thermofield double state of the $\sN=4$ SYM theory above the Hawking-Page temperature should have a stringy horizon at finite (nonzero) $\lam$. 
In contrast, the IOP model~\cite{Iizuka:2008eb} does not, despite exhibiting information loss and type III$_1$ algebras.

\item We introduce an algebra to characterize the structure of a stringy horizon (or a stringy entangling surface), and discuss how to calculate the $O(N^0)$ part of its generalized entropy. 

\item We introduce a modular depth parameter which can be used quantify the  ``depth'' of the emergent radial direction using modular flows of boundary subalgebras. 

This makes it possible to explore an extension of entanglement wedge reconstruction to the stringy regime, and 
gives a diagnostic for the corresponding RT surfaces/QES~\cite{Ryu:2006bv,Hubeny:2007xt,Engelhardt:2014gca}. Using the same techniques as the ones developed in the case of event horizons, we define a QES algebra associated to such stringy quantum extremal surfaces.

We also discuss to what extent the algebraic ER = EPR proposal~\cite{Engelhardt:2023xer} can be extended to the stringy regime, and propose a possible resolution of a question raised in~\cite{Engelhardt:2023xer} about spacetime connectivity for an evaporating black hole before and after the Page time.

\een

The plan of the paper is as follows. In Sec.~\ref{sec:CdpS} we introduce the causal depth parameter and 
give a criterion for the existence of a stringy  horizon. We also discuss some examples. 
In Sec.~\ref{sec:RT} we introduce the modular depth parameter and discuss a diagnostic of stringy RT surface/QES. 
In Sec.~\ref{sec:diss} we conclude with some implications of our results and future perspectives.

\textbf{Notations and conventions:}\begin{itemize} \item 
Single-trace operator algebras associated to a field $\phi$ are denoted $\mathcal{S}_\phi$.
\item Unless mentioned otherwise, in this paper an algebra always means a von Neumann algebra of operators acting explicitly on a Hilbert space.

\item The Fourier transform of a function $f(t)$ is by default denoted by $f(\omega)$. 
\item Without indication, $\int dt \equiv \int_{-\infty}^\infty dt$. \item $\th (\om)$ is the Heaviside step function.
\item In this paper, all Hilbert spaces are separable.
\item $\sB (\sH)$ denotes the algebra of bounded operators on the Hilbert space $\mathcal{H}$
\item  Throughout the text, we assume that all manipulated distributions satisfy the conditions to be tempered. This allows us to take Fourier transforms without ambiguity in our proofs.

\end{itemize}

\section{Causal depth parameter and stringy horizons}\label{sec:CdpS}

In this section we motivate and define a new quantity, which we call \textit{causal depth parameter}, and use it to diagnose the emergence of a causal structure, a radial direction, and horizons in the bulk. We also introduce a horizon algebra to characterize stringy horizons. 
We use $\sN=4$ SYM theory on $S^3$ as the prototypical example, and work in the large $N$ limit throughout. 
For convenience we take the radius of the boundary sphere to be $1$.

\subsection{Causal depth parameter and stringy horizons: the TFD state} \label{sec:hor}

Consider two copies of the boundary CFT in the thermofield double (TFD) state. 
In the large $N$ limit~(for all $\lam$), the system is believed to have a Hawking-Page transition at some temperature $T_{\rm HP}$~\cite{AhaMar03}. For $T < T_{\rm HP}$, the free energy is of order $O(N^0)$. Furthermore, thermal Euclidean two-point functions of single-trace operators are given from the vacuum ones by summing over images in the Euclidean time circle~\cite{BriFes05}, which implies the thermal spectral function has a discrete spectrum. 
For $T > T_{\rm HP}$, the free energy is of order $O(N^2)$, and thermal spectral functions for general single-trace operators are believed to have a complete spectrum~\cite{Festuccia:2006sa}. If we further take the large $\lam$ limit, 
then the boundary system is dual to thermal AdS which consists of two copies of AdS (in an entangled state) for $T < T_{\rm HP}$, and for $T > T_{\rm HP}$ to an eternal black hole~\cite{Mal01}. See Fig.~\ref{fig:TFD}.

The obvious geometric difference between 
thermal AdS and the eternal black hole is that the black hole geometry is connected, with a bifurcate horizon separating the $R$ and $L$ regions of the black hole, whereas two copies of thermal AdS are disconnected. We now would like to generalize the notions of connectedness and horizon to the corresponding bulk dual at finite $\lam$ where we no longer have a geometric description. In order to do so, we will first 
provide an algebraic characterization of these geometric features, after which the generalization to finite $\lam$ is immediate.

\begin{figure}
\centering
\begin{subfigure}{0.4\textwidth}
\centering
\begin{tikzpicture}[scale=0.7]
\node (I)    at ( 4,0)   {};
\node (II)   at (-4,0)   {};
\node (III)  at (0, 2.5) {};
\node (IV)   at (0,-2.5) {};

\path 
   (I) +(90:2)  coordinate[label=45:$t_0$] (It) 
       +(-90:2) coordinate[label=-45:$-t_0$] (Ib) 
       +(180:2) coordinate (Il) 
       ;

\draw[fill=blue!50]  (It) -- (Il) -- (Ib);
\path 
   (I) +(-160:0.8) coordinate[label=:\text{\Large$\fa$}] (Ipt) ;

\path 
   (Il) +(135:1.4) coordinate  (Ilt) 
       +(-135:1.4) coordinate (Ilb) 
       +(180:2) coordinate (M) 
       ;

\draw[pattern=north west lines, pattern color=red]  (Ilt) -- (M) -- (Ilb)-- (Il)-- (Ilt)-- cycle;

\path 
   (Il) +(135:5.63) coordinate  (Ist) 
       +(-135:5.63) coordinate (Isb) 
       ;

\path 
   (I) +(45:0.35)  coordinate (Itop1)
       +(-45:0.35) coordinate (Ibot1)
       +(180:1) coordinate (Ileft1)
       ;
\path 
   (I) +(90:4)  coordinate (Itop)
       +(-90:4) coordinate (Ibot)
       +(180:4) coordinate (Ileft)
       ;

\draw  (Ileft) -- (Itop) -- (Ibot) -- (Ileft) -- cycle;
\path  
  (II) +(90:4)  coordinate  (IItop)
       +(-90:4) coordinate (IIbot)
       +(0:4)   coordinate                  (IIright)
       ;
\draw[decoration=zigzag][fill=red!50] (Ist)--(Ilt)--(M)--(Ilb)--(Isb)decorate{--(IIbot)}--(IItop)decorate{--(Ist)}--cycle;
\draw 
      (IItop) --
          node[midway, below, sloped] {}
      (IIright) -- 
          node[midway, below, sloped] {}
      (IIbot) --
          node[midway, above, sloped] {}
          node[midway, below left]    {}    
      (IItop) -- cycle;
\path 
   (I) +(-174:3.05) coordinate[label=:\text{\Large$\fb$}] (Ipt) ;
\draw[decorate,decoration=zigzag] (Ist) -- (Itop)
      node[midway, above, inner sep=2mm] {};

\draw[decorate,decoration=zigzag] (Isb) -- (Ibot)
      node[midway, below, inner sep=2mm] {};

\end{tikzpicture}
\caption{The Penrose diagram of an AdS-Schwarzschild two-sided black hole.  For all times $t_0$, the bulk algebra $\fa$ dual to the single trace algebra $\mathcal{S}_{(-t_0,t_0)}$ has a nontrivial relative commutant $\fb$ in $\mathcal{S}_{R}$. Topologically, this means that the bulk never fills up and there is a bifurcate horizon.}
\label{fig:TFDa}
\end{subfigure}
~
\begin{subfigure}{0.4\textwidth}
\centering
\includegraphics[scale=0.5]{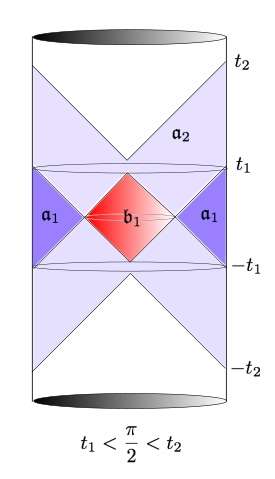}
\caption{A depiction of thermal AdS (right copy). The left copy in the thermofield double state is disconnected and not shown. Contrary to the AdS-Schwarzschild case, the region $\fa_1$ dual to $\mathcal{S}_{(-t_1,t_1)}$ for $t_1<\frac{\pi}{2}$ has a nonempty relative commutant $\fb_1$ inside $\mathcal{S}_{R}$, but for $t_2>\frac{\pi}{2}$, $\mathcal{S}_{(-t_2,t_2)}$ is equivalent to the full right copy, is equal to $\mathcal{S}_{R}$. This means that the bulk ``fills up" for $t_1>\frac{\pi}{2}$, and there is no bifurcate horizon.}
\label{fig:TFDb}
\end{subfigure}
\caption{Diagnosing the presence of a bifurcate horizon from the algebraic structure of boundary time bands.}
\label{fig:TFD}
\end{figure}
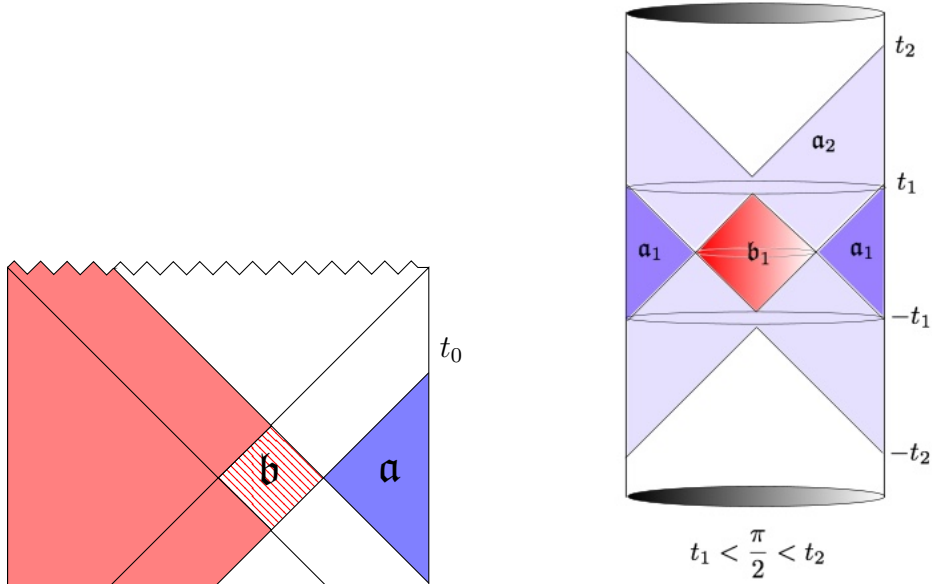

We denote the single-trace operator algebras for the CFT$_{R,L}$ in the large $N$ limit respectively by $\sS_R$ and $\sS_L$. 
They are respectively dual to the bulk regions $R$ and $L$. In particular, in the case of thermal AdS, the bulk $R$-region is simply the right copy of AdS.

Consider an open time interval $I = (-t_0,t_0)$ on the right boundary and the associated single-trace time band operator algebra $\sS_I$. By definition we have $\sS_I \subseteq \sS_R$. 
For the bulk dual given by the black hole geometry, it has been argued that 
$\sS_I$ is equivalent to the bulk algebra in the spherical Rindler region $\fa$ represented on Fig.~\ref{fig:TFD} (a)~\cite{LeuLiu21b,LeuLiu22}.
The existence of a horizon (or more coarsely, the fact that both sides of the black hole are connected) implies the geometric statement that for any choice of $t_0$, $\fa$ only covers a proper subset of the $t=0$ Cauchy slice of the $R$-region of the black hole. In other words, when the two sides of the thermofield double are connected, the $R$-region can never be ``filled up'' by  spherical Rindler regions, no matter how large $t_0$ is. In contrast, consider the copy of thermal AdS represented on Fig.~\ref{fig:TFD} (b). There, $\fa$ covers a full Cauchy slice of the $R$-region for $t_0 \geq {\pi \ov 2}$. 

The above geometric statement can be rephrased algebraically by saying that for the black hole geometry, the algebra $\sS_I$ is a strict subalgebra of $\sS_R$ for any  $t_0$. More explicitly, the commutant $\sS_I'$ is equivalent to the bulk causal complement of $\fa$, which includes the red shaded region and the striped region $\fb$.  The relative commutant of $\sS_I$ in $\sS_R$ is given by $\sS_I' \cap \sS_R$ and is equivalent to the striped bulk region $\fb$. 
This is to be contrasted with the thermal AdS case, where we have $\sS_I' \cap \sS_R = \emptyset$~(or equivalently $\sS_I = \sS_R$) for $t_0 \geq {\pi \ov 2}$.
Such algebraic statements are powerful, as they do not refer to any geometry, and thus can immediately be generalized to diagnostics of the presence or absence of a horizon in the stringy regime at finite $\lam$.

The above discussion motivates the following definition: 

\begin{defn}The two sides of the thermofield double are \textit{stringy-connected} if the algebra $\sS_{(-t_0,t_0)}$ is a strict subalgebra of $\sS_{R}$ for any finite value of $t_0$.\label{def0}\end{defn}

In the case where the two sides of the thermofield double are not stringy-connected, we can also probe the depth of the emergent radial direction thanks to the following:

\begin{defn}\label{def1}
The causal depth parameter $\mathcal{T}$ of a large $N$ thermofield double state is the largest value of $2t_0$ for which the algebra $\sS_{(-t_0,t_0)}$ is a strict subalgebra of $\sS_{R}$.
\end{defn} 

From Definition~\ref{def1}, bulk geometry gives that in the $\lam \to \infty$ limit  
\be\label{ejj}
\sT = \bca {\pi} & T<T_{\rm HP} \cr
\infty & T>T_{\rm HP}
\eca \ .
\ee
In Sec.~\ref{sec:exam}, we will give a purely boundary argument for ~\eqref{ejj}, and show that it generalizes to finite $\lam$. Using Definition~\ref{def1}, Definition \ref{def0} can also be stated by requiring that 
\be 
\sT = \infty   \ .
\ee

Intuitively, in the semi-classical limit, the depth parameter quantifies the depth of the bulk $R$-region, by relating this depth to
the time width of the boundary region obtained by 
sending light rays from a bulk point. See Fig.~\ref{fig:depth}.  
Note that when there is a horizon, clearly by definition, $\sT$ is infinite. In~\cite{Engelhardt:2016vdk}, such considerations were in fact used to characterize the depth of the bulk geometry using the bulk causal structure, and in~\cite{Nebabu:2023iox}, a quantum circuit model following the same line of thought was constructed. Here we have given an intrinsic boundary definition in terms of the commutant structure, which can in turn be used to characterize the bulk causal structure in the stringy regime.

\begin{figure}
\centering
\includegraphics[scale=0.4]{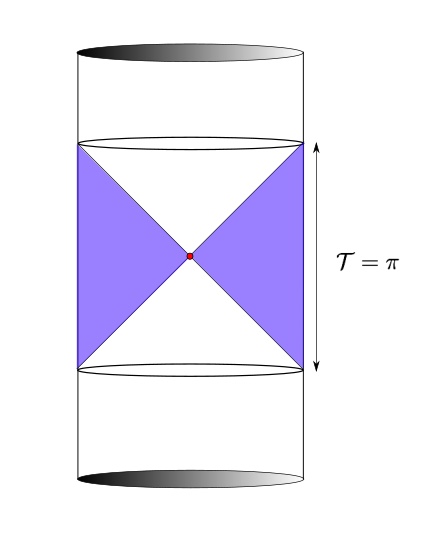}
\caption{A heuristic picture of the depth parameter in thermal AdS. In this case, there is a point, depicted in red, which sits at the ``center" of the bulk. The depth parameter $\mathcal{T}$ can be seen as the boundary time interval that can be reached by shooting light rays from this point. Here, $\mathcal{T}=\pi$.}
\label{fig:depth}
\end{figure}

The depth parameter also has an interesting boundary interpretation. 
It is the minimal width of a time interval such that the knowledge of the algebra of observables localized in the interval is enough to determine the full algebra of the $R$-system.  
For example, in the case of thermal AdS,  $\sS_I = \sS_R$ for $t_0 > {\pi \ov 2}$, i.e., we already have access to all the operations in the $R$-system given such an $\sS_I$. But in the case of a black hole, it is not possible to have access to the full $\sS_R$ for any $t_0$. For the theory at finite $N$, access to the operator algebra on a single Cauchy slice suffices 
to determine its future evolution. But the large $N$ limit, such a property breaks down as there are no equations of motion. This is related to the loss of determinism on the boundary in the large $N$ limit.
Hence there is a sense that the large $N$ limit leads to information loss even in the vacuum sector, although in the end all the information can be recovered by having access to the operator algebra for a time band with finite $t_0 > {\pi \ov 2}$. In the black hole sector, information loss is more drastic: one can never recover all the information even with access to operator algebras of finite time bands of arbitrary width. The causal depth parameter $\mathcal{T}$ can thus also be viewed as quantifying the amount of information loss there is in the system. Whenever it is nonzero, there is a loss of determinism. Under this interpretation, it is not only the formation of black holes, but also the emergence of the whole radial direction, that is related to information loss.
It should also be mentioned that even systems with $\sT=0$ can have information loss, as we will see in an example later.

In the semiclassical regime, while an event horizon is defined through the causal structure, it also has many other properties. For example, it is also a hypersurface of infinite redshift, a property that plays an essential role for the existence of Hawking radiation. In particular, it implies that quantum fields living outside the black hole must have spectral support for arbitrarily small frequencies (in terms of the Schwarzschild time). We would like to have a definition for a stringy black hole that still retains this property. As we will discuss later in Sec.~\ref{sec:genR}, Definition \ref{def0} does not by itself imply that quantum fields living outside such a ``horizon'' have spectral support for arbitrarily small frequencies. 
We will now formulate an additional condition to ensure that. 

An important property associated with a horizon in the semi-classical regime is the existence of not only a bifurcation surface, but also half-sided modular inclusions along the future or past horizon (see Fig.~\ref{fig:half}). In the TFD state, modular flow coincides with time translation. The existence of a half-sided modular inclusion then requires that the algebra of a semi-infinite time band is a strict subalgebra 
of $\sS_R$. This a strictly stronger condition than Definition \ref{def0}, as it involves semi-infinite intervals and not only finite intervals. While it may not be immediately intuitive, we will show in Sec.~\ref{sec:genR} that the emergence of a half-sided inclusion ensures that the spectral support of bulk quantum fields includes arbitrarily small frequencies. While Condition 1 captures some form of connectivity, it is only Condition 2 that guarantees all the properties of horizons. We therefore propose the definition:

\begin{defn}
There is a stringy horizon in the thermofield double state if the subalgebra $\sS_{(0,\infty)}$ is a strict subalgebra of $\sS_{R}$.
\label{def2}
\end{defn}

Note that by invariance under time translation and time reversal, this definition implies that $\sS_{(t_0,\infty)}$ is a strict subalgebra of $\sS_{R}$, and similarly for $\sS_{(-\infty,t_0)}$, for all values of $t_0$.

Instead of imposing the condition of half-sided inclusion, we could imagine defining the existence a stringy horizon by imposing the requirement that the system is ``chaotic''. There are various possible levels of chaos in a system \cite{Gesteau:2023rrx,Ouseph:2023juq}. For example, one possible requirement is that two-point functions cluster in time, i.e., decay with time. In particular, if $\rho(t)$ clusters exponentially fast, i.e., for some constants $C>0$ and $\alpha>0$,
\begin{align}\lvert\rho(t)
\rvert\leq Ce^{-\alpha t},\label{eq:expdecay}
\end{align} 
then $\rho(\omega)$ is analytic on a strip, therefore it can only vanish on isolated points. We will see in Sec. \ref{sec:full} that under mild assumptions, if $\rho(\omega)$ vanishes only at the origin, we can show that there is an emergent half-sided modular inclusion in the system. The condition \eqref{eq:expdecay} gives the weaker condition that $\rho(\omega)$ vanishes at an at most countable number of points on the real axis. It would be interesting to see to what extent a half-sided inclusion also emerges under this weaker assumption.

We also note that by definition $\underset{t_0>0}{\bigvee} \sS_I = \sS_R$, and thus even in the case $\sT = \infty$, we have $\underset{t_0>0}{\bigcap} (\sS_I' \cap \sS_R) = \mathbb{C}$. In particular if we want to define an algebra of operators associated with the bifurcate horizon, it is not enough to consider $\underset{t_0>0}{\bigcap} (\sS_I' \cap \sS_R)$. One instead needs to perform a sequential construction which will be described in Section \ref{sec:alg}.

We can similarly define a depth parameter and horizons using time band subalgebras of $\sS_L$. For the thermofield double state, there is an antiunitary operator $J$ that takes $\sS_R$ to $\sS_L$ and vice versa.  As a result, all the results presented here follow in an exactly analogous manner for $\sS_L$. In particular, as we will see in Sec.~\ref{sec:alg}, there is a precise sense in which the horizon defined using $\sS_L$ coincides with the one defined using $\sS_R$.

\begin{figure}
\centering
\begin{tikzpicture}[scale=0.7]
\node (I)    at ( 4,0)   {};
\node (II)   at (-4,0)   {};
\node (III)  at (0, 2.5) {};
\node (IV)   at (0,-2.5) {};

\path 
   (I) +(90:2)  coordinate (It) 
       +(-90:2) coordinate (Ib) 
       +(180:2) coordinate (Il) 
       ;

\path 
   (Il) +(135:1.4) coordinate  (Ilt) 
       +(-135:1.4) coordinate (Ilb) 
       +(180:2) coordinate (M) 
       ;

\draw[pattern=north west lines, pattern color=red]  (Itop)--(Ilt)--(Ib)--(Itop)--cycle
   (I) +(120:1.9) coordinate[label=:\text{\Large$\mathcal{N}^+$}] (Ipt) ;

\draw[pattern=north east lines, pattern color=blue]  (Ibot)--(Ilb)--(It)--(Ibot)--cycle
   (I) +(-110:2.7) coordinate[label=:\text{\Large$\mathcal{N}^-$}] (Ipt) ;

\path 
   (Il) +(135:5.63) coordinate  (Ist) 
       +(-135:5.63) coordinate (Isb) 
       ;

\path 
   (I) +(45:0.35)  coordinate (Itop1)
       +(-45:0.35) coordinate (Ibot1)
       +(180:1) coordinate (Ileft1)
       ;
\path 
   (I) +(90:4)  coordinate (Itop)
       +(-90:4) coordinate (Ibot)
       +(180:4) coordinate (Ileft)
       ;

\draw  (Ileft) -- (Itop) -- (Ibot) -- (Ileft) -- cycle;
\path  
  (II) +(90:4)  coordinate  (IItop)
       +(-90:4) coordinate (IIbot)
       +(0:4)   coordinate                  (IIright)
       ;
\draw 
      (IItop) --
          node[midway, below, sloped] {}
      (IIright) -- 
          node[midway, below, sloped] {}
      (IIbot) --
          node[midway, above, sloped] {}
          node[midway, below left]    {}    
      (IItop) -- cycle;
\path 
   (I) +(-174:3.05) coordinate (Ipt) ;
\draw[decorate,decoration=zigzag] (IItop) -- (Itop)
      node[midway, above, inner sep=2mm] {};

\draw[decorate,decoration=zigzag] (IIbot) -- (Ibot)
      node[midway, below, inner sep=2mm] {};

\end{tikzpicture}
\caption{Half-sided inclusions for future and past horizons. In the thermofield double state at strong coupling and high temperature, the algebras $\mathcal{N}^+$ and $\mathcal{N}^-$, which supported on future and past semi-infinite time intervals, are dual to the red and blue wedges in the bulk, respectively. In particular, it is the fact that they are inequivalent to the full algebra of right boundary observables that allows for the emergence of nontrivial future and past half-sided modular inclusions, which we will promote to a definition of a horizon in the TFD state in the stringy regime.}
\label{fig:half}
\end{figure}
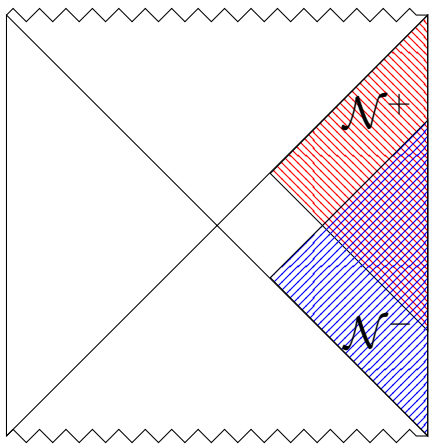

\subsection{Causal depth parameter and stringy horizons for general states} 
\label{sec:hor2}

The definitions of causal depth parameter and stringy horizon in the thermofield double state, as discussed above, can be readily generalized to general semi-classical two-sided and single-sided states. Note, however, that we defer the generalization of the notion of stringy connectivity until Section \ref{sec:RT}.

We refer to a state in the boundary theory as semi-classical if it has  a semi-classical bulk dual in the large $N$ and large $\lam$ limit. We would like to probe the corresponding bulk dual in the stringy regime with a finite $\lam$.

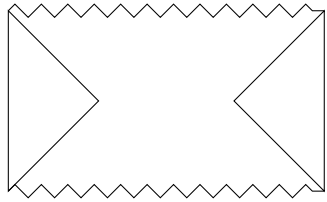
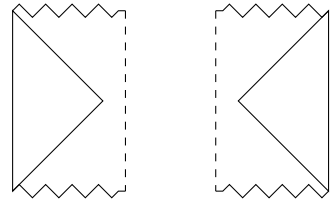
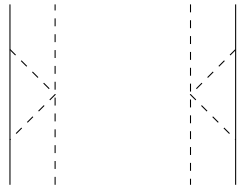
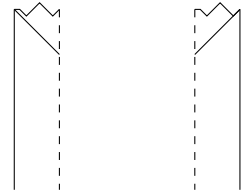
\begin{figure}[htbp]
\centering

\begin{subfigure}[b]{0.4\textwidth}
\centering
\begin{tikzpicture}[scale=0.3]
\node (I)    at ( 7,0)   {};
\node (II)   at (-7,0)   {};
\node (III)  at (0, 2.5) {};
\node (IV)   at (0,-2.5) {};

\path 
   (I) +(45:0.35)  coordinate (Itop1)
       +(-45:0.35) coordinate (Ibot1)
       +(180:1) coordinate (Ileft1)
       ;
\path 
   (I) +(90:4)  coordinate (Itop)
       +(-90:4) coordinate (Ibot)
       +(180:4) coordinate (Ileft)
       ;

\draw  (Ileft) -- (Itop) -- (Ibot) -- (Ileft) -- cycle;
\path  
  (II) +(90:4)  coordinate  (IItop)
       +(-90:4) coordinate (IIbot)
       +(0:4)   coordinate                  (IIright)
       ;
\draw 
      (IItop) --
          node[midway, below, sloped] {}
      (IIright) -- 
          node[midway, below, sloped] {}
      (IIbot) --
          node[midway, above, sloped] {}
          node[midway, below left]    {}    
      (IItop) -- cycle;
\path 
   (I) +(-174:3.05) coordinate (Ipt) ;
\draw[decorate,decoration=zigzag] (IItop) -- (Itop)
      node[midway, above, inner sep=2mm] {};

\draw[decorate,decoration=zigzag] (IIbot) -- (Ibot)
      node[midway, below, inner sep=2mm] {};

\end{tikzpicture}
\caption{The ``long black hole", a connected two-sided geometry.}
\end{subfigure}
\hspace{10pt}
\begin{subfigure}{0.4\textwidth}
\centering
\begin{tikzpicture}[scale=0.3]
\node (I)    at ( 7,0)   {};
\node (II)   at (-7,0)   {};
\node (III)  at (0, 2.5) {};
\node (IV)   at (0,-2.5) {};

\path 
   (I) +(45:0.35)  coordinate (Itop1)
       +(-45:0.35) coordinate (Ibot1)
       +(180:1) coordinate (Ileft1)
       ;
\path 
   (I) +(90:4)  coordinate (Itop)
       +(-90:4) coordinate (Ibot)
       +(180:4) coordinate (Ileft)
       +(180:5) coordinate (Illeft)
       ;
\path (Itop)+(180:5) coordinate (Ittop)
        ;
\path (Ibot)+(180:5) coordinate (Ibbot)
        ;
\path (IItop)+(0:5) coordinate (IIttop)
        ;
\path (IIbot)+(0:5) coordinate (IIbbot)
        ;

\draw  (Ileft) -- (Itop) -- (Ibot) -- (Ileft) -- cycle;
\path  
  (II) +(90:4)  coordinate  (IItop)
       +(-90:4) coordinate (IIbot)
       +(0:4)   coordinate                  (IIright)
       +(0:5) coordinate (Illeft)
       ;
\draw 
      (IItop) --
          node[midway, below, sloped] {}
      (IIright) -- 
          node[midway, below, sloped] {}
      (IIbot) --
          node[midway, above, sloped] {}
          node[midway, below left]    {}    
      (IItop) -- cycle;
\path 
   (I) +(-174:3.05) coordinate (Ipt) ;
\draw[decorate,decoration=zigzag] (Itop) -- (Ittop)
      node[midway, above, inner sep=2mm] {};
\draw[decorate,decoration=zigzag] (IItop) -- (IIttop)
      node[midway, above, inner sep=2mm] {};
\draw[decorate,decoration=zigzag] (Ibot) -- (Ibbot)
      node[midway, above, inner sep=2mm] {};
\draw[decorate,decoration=zigzag] (IIbot) -- (IIbbot)
      node[midway, below, inner sep=2mm] {};
\draw[dashed] (IIttop) -- (IIbbot)
      node[midway, above, inner sep=2mm] {};
\draw[dashed] (Ittop) -- (Ibbot)
      node[midway, above, inner sep=2mm] {};      
\end{tikzpicture}
\caption{A disconnected two-sided geometry with bifurcate horizons.}
\end{subfigure}
\vspace{10pt}
\begin{subfigure}[b]{0.4\textwidth}
\centering
\begin{tikzpicture}[scale=0.3]
\node (I)    at ( 5,0)   {};
\node (II)   at (3,0)   {};
\node (III)  at (-3, 0) {};
\node (IV)   at (-5,0) {};

\path 
   (I) +(90:2)  coordinate (Itop1)
       +(-90:2) coordinate (Ibot1)
       +(180:2) coordinate (Ileft1)
       ;
\path 
   (I) +(90:4)  coordinate (Itop)
       +(-90:4) coordinate (Ibot)
       +(180:4) coordinate (Ileft)
       ;

\path 
   (II) +(90:4)  coordinate (IItop)
       +(-90:4) coordinate (IIbot)
       +(180:4) coordinate (IIleft)
       ;

\path 
   (IV) +(90:2)  coordinate (IVtop1)
       +(-90:2) coordinate (IVbot1)
       +(0:2) coordinate (IVright1)
       ;

\path 
   (III) +(90:4)  coordinate (IIItop)
       +(-90:4) coordinate (IIIbot)
       +(180:4) coordinate (IIIleft)
       ;

\path 
   (III) +(90:4)  coordinate (Itop2)
       +(-90:4) coordinate (Ibot2)
       ;

\path 
  (IV) +(90:4)  coordinate  (IItop2)
       +(-90:4) coordinate (IIbot2)
       +(180:4)   coordinate                  (IIright2)
       ;

\path  
  (IV) +(90:4)  coordinate  (IVtop)
       +(-90:4) coordinate (IVbot)
       +(0:4)   coordinate                  (IVright)
       ;

\draw[dashed] (Itop1) -- (Ileft1) -- (Ibot1);
\draw[dashed] (IVtop1) -- (IVright1) -- (IVbot1);
\draw[dashed] (IItop) -- (IIbot) ;
\draw (Ibot) -- (Itop);
\draw (IVtop) -- (IVbot);
\draw[dashed] (IIIbot) -- (IIItop) ;
\end{tikzpicture}
\caption{A disconnected two-sided geometry without horizons.}
\end{subfigure} 
\hspace{10pt}
\begin{subfigure}[b]{0.4\textwidth}
\centering
\begin{tikzpicture}[scale=0.3]
\node (I)    at ( 5,0)   {};
\node (II)   at (3,0)   {};
\node (III)  at (-3, 0) {};
\node (IV)   at (-5,0) {};

\path 
   (I) +(90:2)  coordinate (Itop1)
       +(-90:2) coordinate (Ibot1)
       +(180:2) coordinate (Ileft1)
       ;
\path 
   (I) +(90:4)  coordinate (Itop)
       +(-90:4) coordinate (Ibot)
       +(180:4) coordinate (Ileft)
       ;

\path 
   (II) +(90:4)  coordinate (IItop)
       +(-90:4) coordinate (IIbot)
       +(180:4) coordinate (IIleft)
       +(90:2)   coordinate                  (IItop1)
       ;

\path 
   (IV) +(90:2)  coordinate (IVtop1)
       +(-90:2) coordinate (IVbot1)
       +(0:2) coordinate (IVright1)
       ;

\path 
   (III) +(90:4)  coordinate (IIItop)
       +(-90:4) coordinate (IIIbot)
       +(180:4) coordinate (IIIleft)
       ;

\path 
   (III) +(90:4)  coordinate (Itop2)
       +(-90:4) coordinate (Ibot2)
       +(90:2)   coordinate                  (IIItop1)
       ;

\path 
  (IV) +(90:4)  coordinate  (IItop2)
       +(-90:4) coordinate (IIbot2)
       +(180:4)   coordinate                  (IIright2)
       
       ;

\path  
  (IV) +(90:4)  coordinate  (IVtop)
       +(-90:4) coordinate (IVbot)
       +(0:4)   coordinate                  (IVright)
       ;

\draw (IItop1) -- (Itop);
\draw (IIItop1) -- (IVtop);
\draw[decoration=zigzag] (Itop) decorate{-- (IItop)};
\draw[decoration=zigzag] (IIItop) decorate{-- (IVtop)};
\draw[dashed] (IItop) -- (IIbot);
\draw (Itop) -- (Ibot);
\draw[dashed] (IIItop) -- (IIIbot);
\draw (IVtop) -- (IVbot);
\end{tikzpicture}
\caption{A disconnected two-sided geometry with future horizons.}
\end{subfigure}
\caption{Cartoons for various types of two-sided states. 
These examples also illustrate that the causal depth parameter does not allow to predict spacetime connectivity in a general state. The cases (a) and (b) both have infinite causal depth parameter, but (a) is connected whereas (b) is disconnected. Similarly, (b) and (c) are both disconnected but (b) has infinite causal depth parameter while (c) has finite causal depth parameter. Finally case (d) is disconnected but has a future horizon.}
\label{fig:genex}
\end{figure}

First consider a general semi-classical two-sided state $\ket{\Psi}$ in the $\lam \to \infty$ limit. See Fig.~\ref{fig:genex} for some cartoon examples. 
Denote by $\sM_R$ the algebra of operators of the $R$-system in the large $N$ limit. 
In the case of TFD state we have $\sS_R= \sM_R$, and time translations coincide with modular flows. 
But in general we have $\sS_R \subset \sM_R$ and the state is not time translation invariant.\footnote{As we will discuss further in Sec.~\ref{sec:RT}, there may be additional operators generated by modular flows. Due to these complications that we will introduce a more general definition of stringy connectivity in Sec.~\ref{sec:RT}.} 
Nevertheless, the definition of causal depth parameter and the diagnostic of a horizon are largely not affected except that there is no symmetry between left and right, and the depth parameter can be time-dependent. 

\begin{defn}\label{def01}
The right causal depth parameter $\mathcal{T}_R (t)$ of a large $N$ two-sided state is the largest value of $2t_0$ for which the algebra $\sS_{(t-t_0,t+t_0)}$ is a strict subalgebra of $\sS_{R}$. 
\end{defn}

\begin{defn} 
A large $N$ two-sided state has a right bifurcate horizon if semi-infinite future and past time band algebras are proper 
subalgebras of $\sS_R$.

\end{defn} 

In fact, it is enough to show that the causal depth parameter  $\sT_R (t)$ is infinite for one value of $t$, since then it is infinite for all values of $t$. Indeed, fix $t^\prime\in\mathbb{R}$. For all values of $t_0$, the interval $(t^\prime-t_0,t^\prime+t_0)$ is included in some interval of the form $(t+T,t-T)$ where we have chosen $T$ large enough. The algebra of that interval is a strict subalgebra since $\sT_R (t)=\infty$, therefore, so is the algebra of $(t^\prime-t_0,t^\prime+t_0)$.

We can also diagnose the future and past horizons as 
\begin{defn}  
A large $N$ two-sided  state has a nontrivial right future horizon after time $\tau$ if the algebra $\sS_{(\tau,\infty)}$ is a strict subalgebra of $\sS_R$. 
\end{defn}

\begin{defn}  
A large $N$ two-sided  state has a nontrivial right past horizon before time $\tau$ if the algebra $\sS_{(-\infty, \tau)}$ is a strict subalgebra of $\sS_R$. 
\end{defn} 

There are parallel definitions for the $L$-system. Note that in the case of a state that is invariant under time translations such as the thermofield double, all three definitions coincide. 
 
When $\sT_R (t)$ is finite, i.e. the $R$-region has a finite depth, the bulk geometry should necessarily be disconnected.
But as indicated in Fig.~\ref{fig:genex}, $\sT_R (t)$ being infinite does not say anything 
whether the bulk geometry is connected or disconnected.  To diagnose whether the bulk geometry is connected requires the further developments of Sec.~\ref{sec:RT} and will be discussed in Sec.~\ref{sec:ER}.

The above definitions can also be straightforwardly adapted for 
a single-sided semi-classical state (i.e. with only one boundary). We just need to change two-sided state to single-sided state and remove the word ``right'' in various places.  For example, 
\begin{defn}\label{def11}
The  causal depth parameter $\mathcal{T} (t)$ of a large $N$ single-sided state is the largest value of $2t_0$ for which the algebra $\sS_{(t-t_0,t+t_0)}$ is a strict subalgebra of $\sS$. 
\end{defn} 
Here $\sS$ is the single-trace operator algebra for the full boundary.

\subsection{The causal depth parameter from the spectral function} \label{sec:spec}

In this subsection we show that the causal depth parameter can be determined by the spectral function of the boundary theory. 
In fact, it is equal\footnote{Up to a technicality.} to 
an invariant known in harmonic analysis as the~\textit{exponential type}.

In the large $N$ limit, in a semi-classical state, correlation functions of single-trace operators factorize into sums of products of two-point functions. Consequently, there is a generalized free field theory around the state for each single-trace operator.
We can diagonalize two-point functions such that different operators do not mix, and thus each generalized free field generates its own algebra. The full algebra is then recovered from all the algebras associated to individual fields. 

In this subsection, we will consider only \textit{one} (Hermitian) generalized free field $\phi$ acting on the GNS Hilbert space of some semi-classical state $\ket{\Psi}$, which can be the TFD state, the vacuum or other two-sided/single-sided states. We will return to the problem of assembling the results obtained for each individual field later in this section.
In the case of a two-sided state, $\phi$ should be understood as $\phi_R$, i.e., a generalized free field in the CFT$_R$. 
For details on the setup, see Appendix~\ref{app:araki}. 
For convenience we will assume $\ket{\Psi}$ is time translation invariant and suppress possible spatial directions for notational simplicity.

Consider the commutator
\be \label{com1}
 \rho (t-t') \equiv  \vev{\Psi|[\phi (t), \phi (t')]| \Psi} 
\ee
which satisfies $\rho^* (t) =  \rho (-t) = - \rho (t)$. The Fourier transform of $\rho (t)$ is called the spectral function and satisfies the properties
\be 
\rho (-\om) = - \rho (\om), \quad \ep (\om) \rho (\om) \geq 0  \ .
\ee
In the TFD state $\rho (\om)$ is related to the Wightman function as 
 \begin{align}
 G_+ (\omega)=\frac{\rho(\omega)}{1-e^{-\beta\omega}} ,
  \end{align}
while in the vacuum state we have 
  \be 
  G_+ (\om) = \th (\om) \rho (\om) \ .
  \ee

The algebra $\sS_\phi$ associated with the field $\phi$ is generated by the
\be 
\phi (f) = \int d t \, f (t) \phi (t)
\ee
with $f(t)$ real and normalizable in terms of the inner product
\be 
\langle f_1, f_2  \rangle_\beta =\int {d\om \ov 2 \pi} \, f^*_1 (\om) \, \th (\om) {\rho (\om)}
\, f_2  (\om) \ .
\label{eq:inner2}
\ee
In the above equation, $f_1 (\om)$ is the Fourier transform of $f_1 (t)$, with $f_1 (-\om) = f_1^* (\om)$. 

For a general state (not equal to vacuum), normalizability for the inner product \eqref{eq:inner2} is too weak. On top of this normalizability, we need to impose the extra condition that $f$ is in the domain of the operator of multiplication by $\mathfrak{b}_-$ defined in Eq. \eqref{c11HH}. In the case of a thermal state which will be of most interest to us, this is equivalent to the stronger normalizability condition \be 
\langle f_1, f_2  \rangle_\beta =\int {d\om \ov 2 \pi} \, f^*_1 (\om) \, \th (\om) \frac{\rho (\om)}{1-e^{-\beta\omega}} 
\, f_2  (\om) \ .
\label{eq:inner3}
\ee

Note that in order to make these inner products positive definite in the case in which the support of $\rho(\omega)$ is not the full real line, the set of the admissible $f(t)$ should also be quotiented out by the functions that are supported only outside the support of $\rho$.

 The subalgebra $\sS_I \subseteq \sS_\phi $ for the time band $I = (-t_0, t_0)$ is generated by $\phi (f_I)$ with the support  ${\rm supp} \, f_I (t) \subset I$. We are interested in whether this inclusion is strict. For a generalized free field, this question can be shown to be equivalent to asking whether the relative commutant of $\sS_I$ in $\sS_R$, i.e., $\sS_I^\prime\cap\sS_R$, is nonempty (see Appendix \ref{app:equ}). The presence of a nontrivial relative commutant indicates the existence of a region in the emergent dual stringy spacetime that is spacelike to the operators in the causal wedge of $I$.

 The (relative) commutant of $\sS_I$ can be shown to be generated by $\phi (g)$, for $g$ real-valued,  satisfying  
\be \label{defg} 
[\phi (g), \phi (f_I)] =0 \quad \to \quad (g, f_I) \equiv \int dt dt' \, g(t) \rho(t-t') f_I (t') = 0  \ .
\ee
Equation~\eqref{defg} by itself only says that the subalgebra generated by $\phi (g)$ belongs to the (relative) commutant of $\sS_I$. 
That it in fact generates the full relative commutant follows from a theorem by Araki, which we review in Appendix~\ref{app:araki}. Equation~\eqref{defg} can also be stated as the fact that $\phi (g)$ lies in the symplectic complement of the subspace spanned by $f_I$, since $(g,f)$, as defined by~\eqref{defg}, is a symplectic product.

Condition~\eqref{defg} can also be phrased as the requirement that the support of the convolution 
\be\label{defs}
(g \ast \rho)  (t') = \int dt \, g(t) \rho (t^\prime-t) ,
\ee
lies in the complement of $I$. 
It is convenient to write~\eqref{defg} and~\eqref{defs} in frequency space 
\be 
\int {d \om  \ov 2 \pi} g^* (\om) \rho (\om) f_I (\om) =0, \quad (g \ast \rho)  (\om) = g (\om) \rho (\om)  \ .
\ee
Thus the relative commutant  is fully determined by the spectral function $\rho (\om)$.

It turns out that our question maps to a well-studied problem in harmonic analysis, known as the \textit{exponential type problem}. The exponential type problem is at the core of some very deep recent results in mathematics and can be formulated as follows:

\textbf{Exponential type problem:} Given a measure $\rho(\omega)$, what is the smallest value of $\sT$ for which the Fourier transforms of compactly supported distributions on $(-\sT/2, \sT/2)$ become dense in $L^2(\rho)$?

The answer to the exponential type problem applied to the spectral density is exactly\footnote{Assuming that in the finite temperature case, the extra constraint of normalizability with respect to the inner product \eqref{eq:inner2} for $\beta$ finite instead of infinite does not change the answer.} the depth parameter $\mathcal{T}$. In Appendix \ref{app:exp}, we review the state of the art on the exponential type problem. While the general case is difficult, it turns out to be quite simple to obtain $\sT$ for various cases of interest for holographic applications, which we will discuss in 
Sec.~\ref{sec:exam}. 

It is amusing to note that the exponential type problem was first introduced in the work of Kolmogorov and others on chaos in classical dynamical systems (see for example \cite{krein}), where they asked how much time is necessary to observe a system for in order to be able to predict its whole evolution. Our discussion provides a holographic interpretation of such a question. In particular, this relationship to dynamical systems is closely related to the notion of information loss in a large $N$ system, which we commented on earlier in Sec.~\ref{sec:hor}.

\subsection{General statements on $\sT$ and half-sided inclusions} \label{sec:genR}

In Sec.~\ref{sec:hor}, we introduced two different notions: 

(i) the notion of stringy connectivity in the thermofield double (for which $\mathcal{T} = \infty$); 

(ii) stringy horizons, in terms of half-sided inclusions of semi-infinite time band algebras. 

\noindent We saw that (ii) implies (i). Here we further clarify the different physics behind these two conditions.
We will see that (i) is only concerned with large frequency behavior of the  boundary generalized free field theory, while (ii) is concerned with both large and small frequency behavior.

The value of the causal depth parameter $\mathcal{T}$ can be obtained  
from the largest value of $a$ for which an \textit{$a$-uniform, or $a$-regular}, sequence can be embedded into the spectral density of the generalized free field under consideration. While we refer the reader to Appendix \ref{app:exp} for a detailed explanation of these notions, here we give an intuitive account of their physical meaning.

Roughly speaking, an $a$-uniform or $a$-regular sequence can be seen as a sequence of points that are embedded in the support of the spectral function, and that become equally spaced with spacing $1/a$ as $\omega\rightarrow\pm\infty$. In other words, it is a sequence $(\lambda_n)$
such that
\begin{align}
\rho (\lam_n) \neq 0, \quad \lam_n \approx \frac{n}{a} + c_{\pm}, \quad n \to \pm\infty,
 \end{align} 
 where the precise meaning of the symbol $\approx$ is clarified in Appendix \ref{app:exp}. In particular, only the behavior at large frequencies matters when deciding whether a sequence is $a$-uniform or not. 
 
Since the depth parameter $\mathcal{T}$ is encoded in the largest value of $a$ for which an $a$-uniform sequence can be embedded in the spectral density, we deduce that the value of $\mathcal{T}$ (and whether or not it is infinite) is solely encoded in the large-frequency behavior of the boundary spectral function. However, we stress that it is not how $\rho (\om)$ grows or decays with a large $\om$---usually characterized as the UV behavior---that is important here. Rather, it is the average spacing between the points in the support of $\rho$ that is required to determine $\mathcal{T}$.

In the semi-classical regime, while an event horizon is defined by the causal structure, it is also a hypersurface of infinite shift. This implies that spectral functions of matter fields outside it are supported at arbitrarily small frequencies, which is in turn important for the thermal interpretation.
The condition $\sT = \infty$ can be viewed as the algebraic counterpart of the causal definition for an event horizon. The above discussion tells us that whether $\sT = \infty$ or not is not sensitive to the behavior of $\rho (\om)$ at small $\om$. For example, $\rho(\om)$ can have a gap near $\om =0$ and still give $\sT = \infty$. This may be interpreted as an indication that,
in the stringy regime, the causal condition no longer warrants the desired spectral behavior needed for interpreting
the horizon as being ``thermal". Thus a stronger condition is needed.

We now show that the condition of half-sided inclusion does require that $\rho (\om)$ is supported at all frequencies.
To see it, we can prove the following proposition:

\begin{prop}
In a (0+1)-D generalized free field theory at finite temperature carrying a half-sided modular inclusion, the spectral function cannot vanish on any open interval.
\label{prop:hsmi}
\end{prop}

\begin{proof}
Suppose that there exists a nonzero function $f$, normalizable for the inner product \eqref{eq:inner2}, in the relative commutant of a half-infinite interval. Then, the Fourier transform of $f(\omega)\rho(\omega)$ must vanish on a half-infinite interval (without loss of generality, say the positive reals). Now we can invoke Beurling's uniqueness theorem \cite{Beurling}, for example the version given in Theorem 4 of \cite{f4312568-16eb-3ae7-9d52-073fcbb92f76}, which we recall here:

\begin{thm}[\cite{f4312568-16eb-3ae7-9d52-073fcbb92f76}]
    Suppose that $T$ is a tempered distribution on $\mathbb{R}$ and that the complement of the support of $T$ contains the disjoint union of closed intervals \begin{align*}\bigcup_{n=1}^\infty[l_n-a_n,l_n+a_n],\end{align*}
    \begin{align*}0<l_1<l_2<\dots<l_n<\dots\to\infty,\end{align*}
    that 
    \begin{align*}\sum_{n=1}^\infty\left(\frac{a_n}{l_n}\right)^2=\infty,\end{align*} and that the Fourier transform $\hat{T}$ vanishes on an open interval. Then, $T=0$.
\end{thm}

In our case we can choose the times $l_n=e^n$ and $a_n=e^{n-1}/2$, to deduce that $f(\omega)\rho(\omega)$ does not vanish on any open interval. In particular, this implies that $\rho(\omega)$ does not vanish on any open interval.
\end{proof}

We now give two physical examples where the depth parameter may be infinite but there is no half-sided modular inclusion.

The spectral function $\rho (\om, \vk)$ of a free scalar field of mass $m$ at zero temperature has the form 
\be \label{frema}
\rho (\om, \vk) = f(-k^2) \th (-k^2 - m^2) (\th (\om)  - \th (-\om)), \quad k^2 \equiv -\om^2 + \vk^2  \ .
\ee
The spectral function has a spectral gap at $\om =0$ for any $\vk$, but has a continuous spectrum beyond the gap. 
The  depth parameter is infinite, but there is no half-sided inclusion, as expected of a system at zero temperature.

Another interesting example is the one of the so-called ``primon gas" or ``Riemannium", whose spectral function is supported at the logarithms of the prime numbers \cite{Julia_1990}. We have:

\begin{prop}
    The depth parameter of the Riemannium satisfies $\mathcal{T}=\infty$, but there is no stringy horizon in the Riemannium at any temperature.
\end{prop}
\begin{proof}
Denote the $n^{th}$ prime number by $p_n$. We construct a subsequence of the $(\mathrm{ln}\,p_n)$ that is $d$-uniform in the sense of Appendix \ref{app:exp} in the following way: consider the sequence $I_n=[n(n+1)/2,(n+1)(n+2)/2]$. The energy condition \eqref{eq:energy} is automatically satisfied. Moreover we can always find $\lvert n\rvert d$ logs of prime numbers in this interval for $\lvert n\rvert$ large enough, as the number of logs of prime numbers in an interval of the form $I_n$ grows exponentially with $n$. Picking such $\lvert n\rvert d$ representatives furnishes a $d$-uniform subsequence. This being true for all $d$, we deduce from Proposition \ref{prop:c7} that $\mathcal{T}=\infty$. However the support of the spectral density is discrete which makes it impossible for there to be a half-sided modular inclusion according to Proposition \ref{prop:hsmi}.
\end{proof}

In contrast to the zero-temperature example~\eqref{frema}, the Riemannium example shows that finite temperature systems can have $\sT = \infty$, but no half-sided inclusion. This exemplifies well the point that $\sT =\infty$ is only concerned with the large-frequency behavior. In particular, the Riemannium does not exhibit enough chaos for there to be a half-sided inclusion. Indeed, it can be seen as an infinite noninteracting superposition of oscillatory modes, with frequencies $\mathrm{ln}\,p$ for all prime numbers $p$. For each of these modes there is no late time decay of the thermal two point function.

\subsection{Examples}\label{sec:exam}

Finding the depth parameter in a generalized free field theory amounts to solving the exponential type problem for the spectral density. While this is difficult in general, in various cases relevant to holography it is possible to deduce whether it is infinite or finite from the qualitative behavior of the spectral function without the need of knowing its explicit form.

Various results are summarized in Table \ref{tab:examples}:

\begin{table}[!ht]
    \centering
    \begingroup
    \renewcommand{\arraystretch}{2.2}
    \scalebox{.8}{\begin{tabular}{cc}
    \toprule
\textbf{Theory} & \textbf{Depth parameter} \\
\midrule
 $\mathcal{N}=4$ SYM, $0\leq T<T_{HP}$ & $\mathcal{T}=\pi$\\
\cmidrule(lr){1-2}
 $\mathcal{N}=4$ SYM at $\lambda>0$, $T>T_{HP}$ & $\mathcal{T}=\infty$ \\
\cmidrule(lr){1-2}
IOP model and its generalizations  at $T > 0$   & $\mathcal{T}=0$\\
\bottomrule
\end{tabular}}
\endgroup
\caption{The depth parameter in the thermofield double state for various theories.}
    \label{tab:examples}
\end{table}

\subsubsection{Evenly spaced spectral density}

We first consider the vacuum state of a $d$-dimensional CFT on $S^{d-1}$. For general $\lam$,  the spectral function of a single-trace operator with given angular quantum numbers on $S^{d-1}$ has the form 
\be \label{vae}
\rho_0 (\om) = 
\sum_{n=0}^\infty a_n \le[\delta(\omega-2n-\omega_0) -\delta(\omega+2n+\omega_0) \ri], 
\ee
for some $\omega_0\in\mathbb{R}$. The above form is determined by the conformal symmetry which also completely fixes all the coefficients $a_n$.

If $\omega_0$ is equal to zero or an integer multiple of $2$, it is be fairly straightforward to show that $\mathcal{T}=\pi$: the Fourier transform of any function multiplied by $\rho(\omega)$ is periodic of period $\pi$, so if it vanishes on an interval of width $\pi$ then it vanishes everywhere, which establishes $\mathcal{T}\leq\pi$. If we further assume the $a_n$ do not decay faster than polynomially, we can further construct explicit multiples of $\rho(\omega)$ whose Fourier transforms vanish on an interval of the form $\left(-\frac{\pi}{2}+\varepsilon,\frac{\pi}{2}-\varepsilon\right)$ for all $\varepsilon>0$, which establishes $\mathcal{T}=\pi$. This construction is detailed in Appendix \ref{app:exp2}.

The general case, however, is more complicated, but the basic idea is that a shift by $\omega_0$ of the spectral function should not alter the result too much, as indicated by our discussion of Sec.~\ref{sec:genR}. This is formalized by a result due to Poltoratski \cite{poltoratski2013problem}, see Appendix \ref{app:exp} for more details and intuition on this result. Since $\rho$ is separated (i.e. there is a minimal spacing between its peaks), and \begin{align}\sum_n\frac{\lvert\mathrm{ln}\,{a_n}\rvert}{1+n^2}<\infty,\end{align} we can apply Theorem 4 of \cite{poltoratski2013problem}. The only value of $a$ for which the counting function $n_\rho$ of $\rho(\omega)$ satisfies \begin{align}\int\frac{n_\rho(x)-ax}{1+x^2}dx<\infty\end{align} is $a=\frac{1}{2}$, so that by Theorem 4 of \cite{poltoratski2013problem}, \begin{align}\mathcal{T}=\frac{2\pi}{2}=\pi.\end{align}

Here the conclusion applies to all single-trace operators, and  thus all bulk stringy fields have the same causal depth.

Theorem 4 of \cite{poltoratski2013problem} is stated in a way that only considers summable spectral functions, however, they still hold for our spectral functions, which have polynomial growth \cite{private}.

Note that technically, the results of \cite{poltoratski2013problem} only guarantee the existence of a \textit{complex-valued} test function $\varphi$ such that $\varphi\ast\rho$ vanishes on an interval of width $\pi$, whereas we need a real-valued one. However, since $\rho$ takes imaginary values, we also have that $\varphi^\ast\ast\rho$ vanishes on an interval of width $\pi$, so $(\varphi+\varphi^\ast)/2$, which is real-valued also vanishes on an interval of width $\pi$.

The above discussion also applies to the TFD state below the Hawking--Page temperature since there the spectral density has the same form~\cite{BriFes05} as well, and the normalizability condition \eqref{eq:inner2} is equivalent to normalizability with respect to $\rho$ (without the temperature-dependent factor) for discretely supported spectral densities. From the above, we deduce that in the vacuum state or in the TFD state below the Hawking-Page transition, the depth parameter is independent of $\lam$. That is, the geometric picture of  empty AdS may extend to finite $\lam$. The statement may not be surprising as empty AdS is the only geometry that is invariant under all the conformal symmetries.

\subsubsection{Compact spectral density}

As our next case, suppose the spectral function $\rho (\om)$ has compact support on the frequency axis. 

\begin{prop}
If the spectral density is compactly supported, then $\mathcal{T}=0$.
\end{prop}

\begin{proof} 
From our discussion of last subsection, to find the (relative) commutant of $\sS_{(-t_0, t_0)}$ we need to identify $g (t)$ such that~\eqref{defs} does
not have support in $(-t_0, t_0)$. 

Since $\rho (\om)$ is compactly supported, so is $g(\om) \rho (\om)$.  
By Fourier inversion for tempered distributions and the Schwartz--Paley--Wiener theorem, $(g\ast \rho) (t)$ is an entire function of $t$.  If it vanishes on an open interval, then it has to be identically zero. We thus conclude the (relative) commutant of $\sS_{(-t_0, t_0)}$ is trivial for any $t_0$. That is, $\sT =0$.
\end{proof}

An explicit example of a compact supported spectral function is the IOP model~\cite{Iizuka:2008eb}, which was proposed as a toy model of black hole information loss. 
It describes an interacting $N$-dimensional vector $a_i$ coupled to a free $N \times N$ matrix. 
The dynamics of the vector is nontrivial, and exhibit information loss in the sense that thermal two-point functions 
of $a_i$ decay to zero at late times. The corresponding spectral function has a continuous spectrum, but is compactly supported. The continuous spectrum means that the von Neumann algebra generated by $a_i$ (of the $R$ system) is type III$_1$~\cite{Derezinski,Golodets_Neshveyev_1998,Furuya:2023fei,Gesteau:2024dhj}. 

$\sT =0$ appears to indicate that there no stringy horizon, nor a sharply defined emergent smooth radial direction. However, we should caution that in the IOP model, vector $a_i$ is only a sector of the full system, in fact a subleading sector.
So it may make sense that this sector does not have a direct bulk geometric interpretation by itself.

There are other matrix quantum mechanical systems in which a probe sector has a spectral function which can be expressed in terms of a semicircle law like in~\cite{Gesteau:2024dhj}, and thus have a compact spectrum. For these systems, we again have 
$\sT=0$, which means no stringy black hole. Results suggestive of a same picture for non-singlets in the $c=1$ matrix quantum mechanics were found in \cite{Betzios:2017yms}.

\subsubsection{Spectral density with full support}\label{sec:full}

Now consider the case of a spectral density $\rho (\om)$  that has a complete support on the real frequency axis. This case was long conjectured~\cite{Festuccia:2005pi,Festuccia:2006sa} to be closely related to the emergence of a stringy bifurcate horizon, here we make this idea precise by showing that a spectral density with full support implies that it allows half-sided modular inclusions in both directions (and hence that the depth parameter is infinite).

\begin{prop}
If the spectral density $\rho (\om)$  is a continuous function that vanishes only at zero, is differentiable at $0$ with continuous nonzero first derivative, and decays at most polynomially, then $\mathcal{T}=\infty$ and there is a stringy horizon.
\end{prop}

\begin{proof}
We need to show that for any $t_0$, it is possible to find $g(t)$ such that $(g \ast \rho) (t)$ does not have support 
in $I =(-\infty, t_0)$ or $I =(t_0, \infty)$.  For this purpose, take a  function $\vp (t)$ that is supported outside $I$. Let
\be \label{defcg}
g(\om) = {\vp (\om) \ov \rho (\om)}, 
\ee
which by construction $g \ast \rho = \vp$ has support outside $I$. To make $g(t)$ real we need 
\be 
g^* (\om) = g (-\om)  \quad \to \quad \vp^* (\om)  = - \vp (-\om)
\ee
where we have used that $\rho (\om)$ is a real, odd function of $\om$. Since $\rho (0) =0$, we also need 
$\vp (0) =0$ such that the ratio~\eqref{defg} is well defined at $\om =0$. Normalizability of $g (\om)$ requires that 
\be 
\int_0^\infty {d \om \ov 2 \pi} {\vp^* (\om) \vp (\om) \ov \rho (\om)} < \infty \ .
\label{eq:norm}
\ee
It is a classical result of Fourier analysis~\cite{Ingham_1934} that given a time band, one can always choose a $\varphi$ with support outside this time band (actually it can even be chosen to have a compact support) such that
\begin{align}
\tilde{\varphi}(\omega)=O(e^{-\lvert\omega\rvert^{1-\alpha}})\label{eq:asymptotics}
\end{align}
for $\alpha>0$. This establishes the normalizability of $g(\omega)$ as long as $\rho$ decays at most polynomially. 
\end{proof}

It is curious to note that the symplectic complement of the union of two disjoint half-infinite intervals may also be nontrivial.
Consider e.g. $I = (-\infty, t_0) \cup (t_1, +\infty)$ with $t_1 > t_0$. We can then choose $\vp (t)$ in the above discussion to have support in $(t_0, t_1)$.

In~\cite{Festuccia:2006sa}, it was argued that for a gauged quantum mechanical system with multiple matrices~(which includes $\mathcal{N} = 4$ SYM theory on $\mathbb{R} \times S^3$) in the large $N$ limit, the spectral function of a generic single-trace operator exhibits a complete spectrum for any nonzero 't Hooft coupling above the Hawking-Page transition temperature.

For large $q$ SYK something more interesting happens: the spectral function has complete spectrum but exponential decay \cite{Maldacena:2016hyu},\footnote{We thank Vladimir Narovlansky for communications on this point.} so equation \eqref{eq:norm} is generically violated for a test function with asymptotics of the form \eqref{eq:asymptotics}. In order to construct a nontrivial relative commutant for a half-infinite time band, we would need to find stronger asymptotic bounds on the decay of $\tilde{\varphi}(\omega)$. The easiest would be for it to have exponential decay, but this is never possible for a test function vanishing on a half-infinite interval (because it is not holomorphic on any strip). Therefore it is not clear whether a relative commutant exists in that case, and if it does its structure is quite different from the one of a spectral density with decay slower than exponential. We leave a more thorough analysis of that case to future work. More generally we believe that our framework could be useful to investigate the SYK model away from the maximally chaotic regime, for example it would be interesting to relate our techniques to the recent results of \cite{Dodelson:2024atp}.

\subsubsection{Systems with uncompact spatial directions}

We now consider some examples of the boundary theory on $\RR^{1,d-1}$.

A simplest example is a free massive field with mass $m$ in the vacuum, whose spectral function was given in~\eqref{frema}. 
 In this case, the spectral function has a gap in the region $|\om| < m$ for all $\vk$. This means that the spectral function has a gap around $\omega=0$. By Proposition \ref{prop:hsmi}, this implies that there is no stringy horizon (in particular, the infinite redshift property is not satisfied). However, the results of \cite{poltoratski2013problem} (extended to polynomially growing spectral spectral density \cite{private}), still imply that the depth parameter $\mathcal{T}$ is infinite.

For a CFT in the vacuum state, the spectral function of a single-trace operator has the form 
\be 
\rho(\om, \vk) = C (-k^2)^{\nu} \th (-k^2) (\th (\om)  - \th (-\om)) , \quad k^2 \equiv -\om^2 + \vk^2  \ .
\ee
For each given $\vk$, there is a gap in the spectrum for $\om^2 < \vk^2$. The gap vanishes in the limit $\vk \to 0$. 
In the case, the bulk is given by the Poincare patch of empty AdS. There is a Poincare horizon there related to the vanishing of the gap in the $\vk \to 0$ limit.

In AdS-Rindler, the spectral function at fixed momentum has the same form as the one of BTZ, except now there is a continuum of momenta. Therefore the depth parameter is infinite, as expected because the bulk is semiclassical.

Consider the spectral function of a scalar glueball in free Yang--Mills theory at finite temperature. It is found 
in~\cite{Hartnoll:2005ju} that the spectral function is continuous and does not vanish outside of an interval. Therefore, following \cite{poltoratski2013problem}, the depth parameter should be infinite and there is a half-sided inclusion.

Another interesting case is the one of the large $N$ limit of symmetric product orbifolds. The bulk theory is thought to be far from regular general relativity on a fixed spacetime, however we were made aware of the upcoming work \cite{orbifold}, which shows that at least for some specific probes, spectral densities are the ones of the BTZ black hole, suggesting an infinite depth parameter and an emergent stringy horizon. We believe understanding this case in more detail could be particularly insightful.

\subsection{Characterization of stringy connectivity: Einstein--Rosen algebra}\label{sec:alg}

We now focus the case of the thermofield double when $\mathcal{T}=\infty$. In that case, there is stringy connectivity between two sides. As we discussed earlier, with $\sT = \infty$ alone, what connects the left and right systems may not qualify as a bifurcate horizon. We will thus refer to it as a stringy ``Einstein-Rosen bridge''. In this section, we will construct an algebra of observables associated with an Einstein-Rosen bridge.

We will need the extra assumption that to each open interval $I$ of $\mathbb{R}$ one can associate not only a von Neumann algebra $M_I$, but also a weak operator dense $C^\ast$-algebra $\mathcal{A}_I$ (this is clearly true in the case of generalized free fields, where one can simply consider the Weyl $C^\ast$-algebra of operators constructed out of the real symplectic structure on the one particle Hilbert space). Denoting $\mathcal{A}=\mathcal{A}_{\mathbb{R}}$, we also assume that $\bigcup_t\mathcal{A}_{(-t,t)}$ is norm-dense in $\mathcal{A}$. We now define a notion of \textit{Einstein--Rosen sequence} (ER sequence), which will allow us to define an algebra attached to a stringy Einstein-Rosen bridge. We will call this algebra an Einstein-Rosen algebra (ER algebra) or with a slight abuse of language, horizon algebra.

\begin{defn}
Let $\mathcal{A}$ be the $C^*$-algebra generated by single-trace operators of one boundary and let $(A_n)$ be a sequence of operators in $\mathcal{A}$. $(A_n)$ is said to be an \textit{ER sequence} if it is bounded in norm, and for all $t\in\mathbb{R}$, $A_n\in \mathcal{A}_{(-t,t)}^\prime$ except for a finite number of terms.
\end{defn}

In order to define a good notion of ER algebra, we first need to introduce one piece of machinery known as ultrapowers and central sequences. 

\begin{defn}
The ultrapower algebra $\mathcal{A}^\omega$ of a $C^\ast$-algebra $\mathcal{A}$ at a free ultrafilter $\omega$ is defined to be the quotient of the space of bounded sequences valued in $\mathcal{A}$ by the space of sequences valued in $\mathcal{A}$ that converge towards a central element (in norm).
\end{defn}

\begin{defn}
A central sequence $(X_n)$ in a $C^\ast$-algebra $\mathcal{A}$ is a bounded sequence of operators $(X_n)$ satisfying, for all $A\in\mathcal{A}$,
\begin{align}
\|[X_n,A]\|\underset{n\rightarrow\infty}\longrightarrow 0.
\end{align}
\end{defn}

We can now first show:

\begin{prop}
ER sequences are central sequences.
\end{prop}

\begin{proof}
Let $A\in\mathcal{A}$, let $\varepsilon>0$. There exists $T\in\mathbb{R}$ and $A_T\in\mathcal{A}_{(-T,T)}$ such that\begin{align}\|A-A_T\|<\varepsilon.\end{align} Moreover for $n$ large enough $A_n$ commutes with $A_T$, so that \begin{align}\|[A_n,A]\|\leq 2\varepsilon\|A_n\|.\end{align}
Since $(A_n)$ is bounded in norm we conclude that $(A_n)$ is a central sequence.
\end{proof}

Of course, not all central sequences are ER sequences! For example, a sequence of modular flowed operators in a strongly coupled large $N$ CFT would be central, but not an ER sequence.

We now define the ER algebra as

\begin{defn}
Let $\omega$ be a free ultrafilter. The ER algebra is the ($C^\ast$-closure of) the algebra of ER sequences inside $\mathcal{A}^\omega$, up to trivial ones.
\end{defn}

By the above proposition, the ER algebra is a $C^\ast$-subalgebra of the $C^\ast$-algebra of central sequences $\mathcal{A}_{\omega}=\mathcal{A}^\prime\cap \mathcal{A}^{\omega}$. In principle, it depends on a choice of free ultrafilter. 

The reader can reasonably wonder at this stage why we chose to define the ER algebra at the $C^\ast$-level rather than at the von Neumann algebra level. The reason is that the analogous notion of central sequence at the von Neumann level which is known as the \textit{asymptotic centralizer}, is not compatible with the physics of our setup. More precisely, given a normal state $\varphi$ on a von Neumann algebra $M$, a sequence $(X_n)$ of observables in a $M$ is said to be part of the asymptotic centralizer of $\varphi$ if \begin{align}\|\varphi(X_n\cdot)-\varphi(\cdot X_n)\|_{M_\ast}\underset{n\rightarrow\infty}{\rightarrow}0.\label{eq:asympcentr}\end{align}The convergence in \eqref{eq:asympcentr} must happen for the norm topology in the Banach space of linear functionals on $M$. What this means is that the rate of convergence of $\varphi(X_n A)-\varphi(A X_n)$ towards zero must be \textit{independent} of the choice of $A$. But this is clearly not the case for a sequence $X_n$ of operators whose support approaches ER bridge: the closer $A$ is to it, the longer it will take for $\varphi(X_n A)-\varphi(A X_n)$ to converge towards zero. Therefore we need a notion of central sequence that allows for dependence on the operator $A$, and this is provided by the definition at the $C^\ast$-level. The price to pay is that the ER algebra is only made of sequences of operators that are in the $C^\ast$-closure of the space of local observables, rather than the von Neumann closure which is larger.

Another question is whether the left and right ER algebras detect ``the same geometric object". One way to check this is to ask whether the commutant of the algebra generated by two boundary time bands shrinks to a trivial algebra as the width of the time bands goes to infinity. If it does, then there are no observables between the right and left ER bridges and they coincide. See Figure \ref{fig:qesalg}. We expect this to be generically true in the thermofield double state, where the full algebra of a single boundary in the large $N$ limit is generated by single-trace operators.  
 In the next section, we will see that in more general states however, there may be more large $N$ operators, such as modular flowed ones, and that the presence of these large $N$ operators can alter the definition of stringy connectivity.

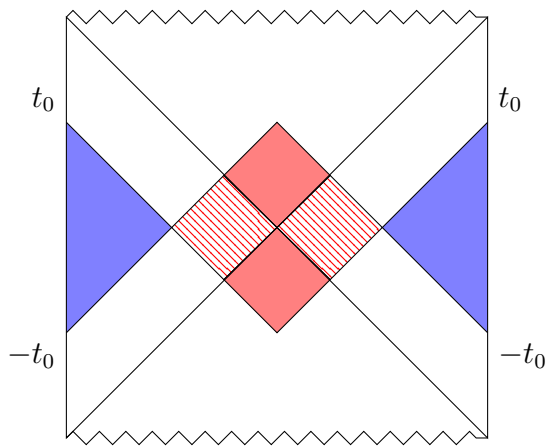
\begin{figure}
\centering
\begin{tikzpicture}[scale=0.7]
\node (I)    at ( 4,0)   {};
\node (II)   at (-4,0)   {};
\node (III)  at (0, 2.5) {};
\node (IV)   at (0,-2.5) {};

\path 
   (I) +(90:2)  coordinate[label=45:$t_0$] (It) 
       +(-90:2) coordinate[label=-45:$-t_0$] (Ib) 
       +(180:2) coordinate (Il) 
       ;

\draw[fill=blue!50]  (It) -- (Il) -- (Ib);
\path 
   (I) +(-160:0.8) coordinate (Ipt) ;

\path 
   (Il) +(135:1.4) coordinate  (Ilt) 
       +(-135:1.4) coordinate (Ilb) 
       +(180:2) coordinate (M) 
       ;

\path
(M) +(90:2) coordinate (M1)
+(-90:2) coordinate (M2)
+(180:2) coordinate (Ir)
;

\path
(Ir) +(45:1.4) coordinate  (Irt) 
       +(-45:1.4) coordinate (Irb) ;

\path 
   (Il) +(135:5.63) coordinate  (Ist) 
       +(-135:5.63) coordinate (Isb) 
       ;

\path 
   (I) +(45:0.35)  coordinate (Itop1)
       +(-45:0.35) coordinate (Ibot1)
       +(180:1) coordinate (Ileft1)
       ;
\path 
   (I) +(90:4)  coordinate (Itop)
       +(-90:4) coordinate (Ibot)
       +(180:4) coordinate (Ileft)
       ;

\path 
   (II) +(90:2)  coordinate[label=135:$t_0$] (IIt) 
       +(-90:2) coordinate[label=-135:$-t_0$] (IIb) 
       +(0:2) coordinate (IIl) 
       ;
\draw[fill=blue!50]  (IIt) -- (IIl) -- (IIb);

\draw[pattern=north west lines, pattern color=red]  (Ilt) -- (M) -- (Ilb)-- (Il)-- (Ilt)-- cycle;
\draw[pattern=north west lines, pattern color=red]  (Irt) -- (M) -- (Irb)-- (Ir)-- (Irt)-- cycle;
\draw[fill=red!50]  (Irt) -- (M1) -- (Ilt)-- (M)-- (Irt)-- cycle;
\draw[fill=red!50]  (Irb) -- (M2) -- (Ilb)-- (M)-- (Irb)-- cycle;
\draw  (Ileft) -- (Itop) -- (Ibot) -- (Ileft) -- cycle;
\path  
  (II) +(90:4)  coordinate  (IItop)
       +(-90:4) coordinate (IIbot)
       +(0:4)   coordinate                  (IIright)
       ;
\draw 
      (IItop) --
          node[midway, below, sloped] {}
      (IIright) -- 
          node[midway, below, sloped] {}
      (IIbot) --
          node[midway, above, sloped] {}
          node[midway, below left]    {}    
      (IItop) -- cycle;
\path 
   (I) +(-174:3.05) coordinate (Ipt) ;
\draw[decorate,decoration=zigzag] (IItop) -- (Itop)
      node[midway, above, inner sep=2mm] {};

\draw[decorate,decoration=zigzag] (IIbot) -- (Ibot)
      node[midway, below, inner sep=2mm] {};

\end{tikzpicture}
\caption{Equivalence between the right and left definitions of the ER algebra in the thermofield double state. The holographic duals of the left and right time bands, represented in blue, are related to each other by modular conjugation. The relative commutants of each time band inside the left and right algebras, each carrying red dashed lines, generate the commutant of the two algebras, filled in red. As $t_0\rightarrow\infty$, the three diamonds all collapse onto the bifurcate horizon.}
\label{fig:qesalg}
\end{figure}

\subsection{Generalized entropy of a stringy horizon}

So far, our framework allows to detect whether or not there is an emergent stringy bifurcate horizon in the bulk, but it does not allow yet to talk about the ``area" of this horizon, or any related measure that it could be equipped with. A full resolution of this problem is beyond the scope of this paper, but here we make the remark that it is at least possible to define the difference between the generalized entropy of an excited state of the CFT and the general entropy of the thermofield double. Indeed, by Wall's proof of the generalized second law \cite{Wall:2011hj}, we know that at strong coupling, if $\ket{\psi}$ is an excitation of the TFD state $\ket{\rm TFD}$, the difference in generalized entropies between the entropy of the full bulk subregion and the entropy at infinity is given by 
\begin{align}
S_{gen}(\infty)-S_{gen}(b)=S(\psi\vert \rm TFD),
\end{align}
where the right hand side is the relative entropy. It is therefore natural to promote this formula to a \textit{definition} of the difference between the generalized entropy of our state for the full subregion and its generalized entropy at infinity. 

In the case in which the state $\ket{\psi}$ is a coherent excitation of $\ket{\rm TFD}$ of the form $W(f)\ket{\rm TFD}$ in a $(0+1)$-D GFF theory, where $f$ is in the one particle Hilbert space and $W(f)$ is the corresponding Weyl operator, there is an explicit formula due to \cite{Ciolli:2019mjo,Bostelmann_2021}, which 
gives:
\begin{align}S(\psi\vert \mathrm{TFD})=\int d\omega\theta(\omega)\frac{\omega}{1-e^{-\beta\omega}}\vert f(\omega)\vert^2. \end{align}

\subsection{Violations of the equivalence principle in the stringy regime?}

As discussed in the Introduction section, in the stringy regime, different bulk fields may ``see'' different bulk geometries. 
The causal depth parameter $\sT$ can be used to probe this phenomenon. 

In a large $N$ gauge theory like $\mathcal{N}=4$ SYM, there are infinitely many generalized free fields. To each of these fields, one can associate a depth parameter. If there is an emergent semiclassical bulk geometry, the large $N$ spectral densities must somehow conspire so that the depth parameter associated to each large $N$ field is the same. Any pair of fields whose depth parameters $\mathcal{T}_1$ and $\mathcal{T}_2$ are not equal cannot be seen as living on the same bulk geometry, as the corresponding depths are different. Therefore such a pair of fields can be seen as giving an obstruction to the validity of the equivalence principle in the bulk.

In the case of empty AdS or thermal AdS, there is a quite general argument ensuring that $\mathcal{T}$ is the same for all fields, as we have seen in Sec.~\ref{sec:exam}. This means that there is a somehow universal notion of ``bulk depth" at large $N$ in the vacuum and  a low temperature thermofield double state, even at finite 't Hooft coupling. However, going away from the vacuum or the thermofield double, one can imagine more general backgrounds in which the equivalence principle may be violated.

For the thermofield double state above the Hawking-Page temperature, we have argued a bifurcate horizon exists for all generalized fields. However, do horizons for different fields coincide? At the moment, we do not have a direct way to answer this question, and will just make some general remarks. The thermofield double state has a left-right reflection symmetry, which enables us to establish in Sec.~\ref{sec:alg} that the left and right horizons for each field coincide. Therefore, if there is a notion of stringy geometry for all fields (although the effective metric may be different for different fields), then the horizons should sit in the middle of that spacetime, and thus coincide. 

Another possible avenue for probing the coincidence of the horizons is to examine how the notion of a common horizon for all fields may become blurry due to potential failures of the split property for algebras describing all bulk fields at finite string length. Indeed, the presence of higher spin fields in the bulk theory leads to nonlocalities that may be related to the failure of the split property \cite{Faulkner:2022ada}, or of the closely related modular nuclearity condition \cite{buchholz2004modular}. Perhaps relatedly, it is known that when introducing interactions between the large $N$ fields, the notion of entanglement wedge becomes blurry in the stringy regime where the out-of-time-ordered correlator gets a submaximal Lyapunov exponent \cite{Chandrasekaran:2021tkb}. 

Finally, as it will be discussed more in the last section, given a spectral function, we can in principle reconstruct the bulk geometry using methods from inverse scattering theory. That may enable a direct comparison.

\section{Algebraic diagnostic of QES} \label{sec:RT}

The definition of causal depth parameter and the algebraic diagnostic of an event horizon discussed in the last section 
were motivated from the bulk causal structure in the semi-classical gravity regime. We then used that data as inspiration to probe/define the bulk causal structure in the stringy regime.
In this section we would like to generalize the 
discussion of the previous section by introducing an algebraic diagnostic for the existence of Ryu-Takayanagi (RT) surfaces~\cite{Ryu:2006bv,Hubeny:2007xt} or more generally quantum extremal surfaces (QES)~\cite{Engelhardt:2014gca}. 
  
The question becomes more difficult as QES are not directly related to the bulk causal structure, which makes them harder to probe. However, the case of the thermofield double, where the bifurcate horizon coincides with the RT surface for the $R$ system, suggests a path forward. 

For the thermofield double, the procedure discussed in Sec.~\ref{sec:hor} 
can also be viewed as diagnosing the possible existence of a particular asymptotically invariant submanifold under the modular flows of the boundary algebra $\sM_R = \sS_R$ for the CFT$_R$, represented with a red dot on Fig. \ref{fig:modularbh}.\footnote{It is an asymptotically fixed submanifold as operators localized on the bifurcate horizon do not belong to  the boundary algebra $\sM_R = \sS_R$ for the CFT$_R$. In fact, since operators always need to be smeared in the time direction, there are no well defined operators localized on the bifurcate horizon.} This manifold has the following properties in terms of modular data: 

\ben
\item  The Schwarzschild time flow becomes a boost near the bifurcation surface, which leaves the surface invariant.

\item  The Schwarzschild time flow is identified with the boundary time flow, which is in turn identified with the modular flows.  

\item While both the red and blue dots of Fig.~\ref{fig:modularbh} are asymptotically invariant submanifolds of the boundary time flows, the ER sequences defined in Sec.~\ref{sec:alg} from the relative commutants of time bands only approach the red dot. Recall that given any finite time band, except for a finite number of terms, an ER sequence $\{A_n\}$ lies in its commutant. In contrast, this is not the case for any sequence approaching the blue dots.

\een

This perspective can now be used to diagnose the possible existence of a QES by {\it defining} a QES as an asymptotically invariant submanifold under the modular flows of the corresponding boundary algebra, in the sense captured by the three properties above.

We can thus apply the constructions of Sec.~\ref{sec:hor} to construct a good notion of QES in the stringy regime by simply replacing time bands by  {\it modular time} bands.  Below we first review aspects of the algebraic formulation of entanglement wedge reconstruction of~\cite{LeuLiu22} that will be used in our discussion. We then present our proposal for an algebraic diagnostic of QES that can also be applied to the stringy regime, and discuss some examples. In particular, the proposal gives a possible solution to a question raised in~\cite{Engelhardt:2023xer} concerning evaporating black holes. 
We also make comments on a possible generalization of the algebraic notion of ER=EPR discussed in~\cite{Engelhardt:2023xer} to the stringy regime.

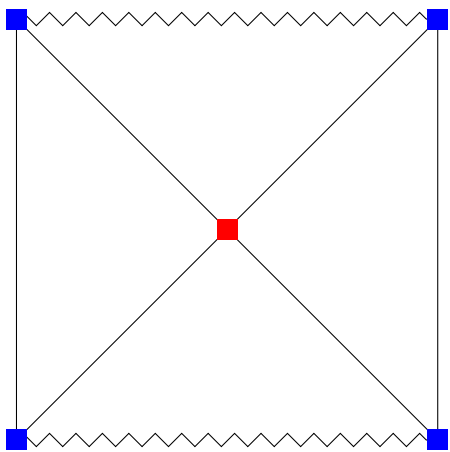
\begin{figure}
\centering

\begin{tikzpicture}[scale=0.7]
\node (I)    at ( 4,0)   {};
\node (II)   at (-4,0)   {};
\node (III)  at (0, 2.5) {};
\node (IV)   at (0,-2.5) {};
\path 
   (I) +(90:2)  coordinate (It) 
       +(-90:2) coordinate (Ib) 
       +(180:2) coordinate (Il) 
       ;

\path 
   (I) +(-160:0.8) coordinate (Ipt) ;

\path 
   (Il) +(135:1.4) coordinate  (Ilt) 
       +(-135:1.4) coordinate (Ilb) 
       +(180:2) coordinate (M) 
       ;

\path 
   (Il) +(135:5.63) coordinate  (Ist) 
       +(-135:5.63) coordinate (Isb) 
       ;

\path 
   (I) +(45:0.35)  coordinate (Itop1)
       +(-45:0.35) coordinate (Ibot1)
       +(180:1) coordinate (Ileft1)
       ;
\path 
   (I) +(90:4)  coordinate (Itop)
       +(-90:4) coordinate (Ibot)
       +(180:4) coordinate (Ileft)
       ;

\draw  (Ileft) -- (Itop) -- (Ibot) -- (Ileft) -- cycle;
\path  
  (II) +(90:4)  coordinate  (IItop)
       +(-90:4) coordinate (IIbot)
       +(0:4)   coordinate                  (IIright)
       ;
\draw 
      (IItop) --
          node[midway, below, sloped] {}
      (IIright) -- 
          node[midway, below, sloped] {}
      (IIbot) --
          node[midway, above, sloped] {}
          node[midway, below left]    {}    
      (IItop) -- cycle;
\path 
   (I) +(-174:3.05) coordinate (Ipt) ;
\draw[decorate,decoration=zigzag] (IItop) -- (Itop)
      node[midway, above, inner sep=2mm] {};

\draw[decorate,decoration=zigzag] (IIbot) -- (Ibot)
      node[midway, below, inner sep=2mm] {};
\node[fill=blue]    at ( 4,4)   {};
\node[fill=blue]    at ( -4,-4)   {};
\node[fill=blue]     at ( -4,4)   {};
\node[fill=blue]     at ( 4,-4)   {};
\node[fill=red]    at ( 0,0)   {};
\end{tikzpicture}
\caption{Fixed points of modular flow in the Penrose diagram of an AdS-Schwarzschild two-sided black hole. The red dot denotes the QES and the blue dots denote timelike infinity on the boundary. The algebraic difference between the red dot and the blue dots is that the red dot is inside the commutant of all double-sided time band algebras, which is not the case for the blue dots.}
\label{fig:modularbh}
\end{figure}

\subsection{Modular depth parameter and QES from modular time bands}

Before stating our proposal for diagnosing the presence of a QES for a boundary subregion algebraically, we first review some elements 
of the algebraic formulation of entanglement wedge reconstruction that are essential for our discussion.

For a system in a semi-classical state $\ket{\Psi}$ describing some bulk geometry, entanglement wedge reconstruction says that physics in the entanglement wedge 
of a boundary {\it spatial} subregion $A$ can be fully described by observables in $A$. The entanglement wedge is given by the domain of dependence $\hat \fb_A$ of a bulk co-dimension one region $\fb_A$ satisfying $\p \fb_A = \ga_A \cup A$ where $\ga_A$ denotes the QES for $A$.
On the boundary, denote the large $N$ limit of the local operator algebra $\sB_A$
in $A$ by $\sX_A$. Entanglement wedge reconstruction can be stated algebraically as 
\be \label{ewr}
\widetilde \sM_{\fb_A} = \sX_A
\ee
where $\widetilde \sM_{\fb_A}$ is the bulk operator algebra for $\fb_A$. Equation~\eqref{ewr} means that any bulk operations in $\hat \fb_A$ can be expressed in terms of those in $\sX_A$.  
We assume that $\ket{\Psi}$ is cyclic and separating with respect to $\sB_A$, and denote the corresponding modular operator by $\De$.

Some of our statements will be conditioned on the following assumptions, which all hold in the semiclassical gravity regime: 

\ben 

\item The QES $\ga_A$ is an asymptotically invariant submanifold under modular flows generated by $\De$.  

This statement comes from the expectation that near $\ga_A$  modular flows generated by $\De$ should act as local boosts.  

\item $\sX_A$ can be generated by  single-trace operators in $A$ via modular flows. More explicitly, $\sX_A$ is generated by large $N$ limits of operators of the form 
\be \label{eow}
\phi (s, \vx) = \De^{-i s} \phi (\vx) \De^{is}, \quad s \in \RR, \; \vx \in A 
\ee
where $\phi$ is a single-trace operator. This statement follows from the arguments of~\cite{Jafferis:2015del, Faulkner:2017vdd} in the strongly coupled regime. 
\item Haag duality survives the large $N$ limit, i.e. \begin{align}\sX_A^\prime=\sX_{\bar{A}}.\end{align} This is guaranteed in the strongly coupled regime by entanglement wedge reconstruction/complementary recovery \cite{Dong:2016eik}.
\een

Note that while we will use assumption 1 essentially as a definition for our notion of stringy QES, assumptions 2 and 3 are not guaranteed to hold true in the stringy regime. In particular, assumption 2 not being true would signal the existence of more exotic large $N$ observables than large $N$ modular flows, while a breakdown of assumption 3 would signal a breakdown of complementary recovery. We leave a more thorough analysis of these assumptions to the future.

Consider an open {\it modular} time interval $I = (-s_0,s_0)$, and denote the subalgebra 
generated by $\phi (s, \vx)$ of~\eqref{eow} with $s \in I$ as $\sX_{I}$. By definition, 
$\sX_I \subseteq \sX_A$. We can then define the modular depth parameter associated with $A$ as follows. 

\begin{defn}\label{def31}
The modular depth parameter $\mathcal{T}_m (A)$ associated with a boundary region $A$ is the largest value of $s_0$ for which the algebra $\sX_{(-s_0,s_0)}$ has a nontrivial relative commutant inside $\sX_{A}$.
\end{defn} 

We now would like to characterize the existence of a nonempty QES at the boundary of the entanglement wedge. The existence of such a nonempty QES can be understood as a statement of connectivity between the entanglement wedge and its complement. Therefore it makes sense to reproduce the definition of stringy connectivity proposed in Section \ref{sec:hor}:

\begin{defn}\label{defcm}
Suppose $R$ and $L$ systems are in an entangled pure state with a large $N$ limit. We say that $R$ and $L$ are share a non-empty QES, or are stringy-connected, if the corresponding modular depth parameters are infinite. 
\end{defn} 

Note that although this definition guarantees that the $R$ and $L$ systems have a QES, it is not clear whether this QES is the same by only making assumption 1 above. However, by adding assumptions 2 and 3 (i.e. that entanglement wedges can be generated by large $N$ limits of modular flows and that Haag duality holds at large $N$ in the bulk), one can show that there are no operators localized in between the right and the left QES, which means that they ``coincide" in this sense. More formally:
\begin{prop}  
If $R$ and $L$ systems have a non-empty QES, and assumptions 1, 2 and 3 are satisfied, then \begin{align}\bigcap_{I}(\sX_I^L\vee\sX_I^R)^\prime=\mathbb{C}.\end{align} 
\end{prop} 

\begin{proof}
\begin{align}
\bigvee_{I}\left(\sX_I^L\vee\sX_I^R\right)=\sX^L\vee\sX^R=\mathcal{B}(\mathcal{H}),
\end{align}
so by going to the complement,
\begin{align}\bigcap_{I}(\sX_I^L\vee\sX_I^R)^\prime=\mathbb{C}.\end{align} 
\end{proof}

The discussion of Sec.~\ref{sec:spec} can also be directly carried over by replacing $\phi (t)$ there by $\phi (s)$ of~\eqref{eow}. The modular depth parameter and possible existence of a nontrivial QES can then be deduced from the behavior of  the modular spectral function in a similar manner to the previous discussion. Note that just as before, this definition staightforwardly extends to the stringy regime, where potentially, different fields can see different entanglement wedges and quantum extremal surfaces. This could lead to at best, field-dependence, and at worst, an inconsistency in the definition of the entanglement wedge, and would be somewhat in line with hints \cite{Chandrasekaran:2021tkb} suggesting difficulties in defining the entanglement wedge universally in the stringy regime. We look forward to investigating the potential field-dependence of stringy entanglement wedges in future work.

Now consider some simple examples. 

Suppose $A$ is a proper subregion on a boundary Cauchy slice. The entanglement wedge is a proper subregion in the bulk and has a nontrivial QES $\ga_A$. Given that modular flows are expected to act as local boosts near $\ga_A$, we deduce that the ``single-particle'' modular spectral function should have a complete spectrum, and from the discussion of Sec.~\ref{sec:full}, the modular depth parameter is $\infty$. We can generalize this discussion to the stringy regime. Suppose the modular depth parameter for $A$ remains infinite at finite $\lam$, we can then say the QES survives finite $\apr$ corrections.

Now consider a two-sided semi-classical state corresponding to the long wormhole shown on Figure \ref{fig:genex} (a). For simplicity, we will suppose that the bulk spacetime is time-reflection symmetric, and consider $A$ to be the full right boundary on the time-reflection symmetric Cauchy slice.  By using nested time bands of single-trace operator algebras defined in Sec.~\ref{sec:hor2}, we can probe the existence of the causal horizon. In this case there is a nontrivial QES lying outside the horizon, which can be probed by nested modular time bands. So in this case both the causal depth parameter $\sT_c$ and the modular depth parameter $\sT_m$ are infinite. The fact that the QES lies outside the horizon is reflected in the fact that $\sS_R$ has a nontrivial relative commutant inside $\sX_R$. Thus the larger the algebra, the deeper we can probe into the bulk, even though we cannot directly compare $\sT_c$ and $\sT_m$, and both are infinite.  At finite $\lam$, we can also use the time bands of $\sS_R$ and modular time bands of $\sX_R$ to define the stringy horizons and QES. That one probes ``deeper'' than the other again follows from the proper inclusion of one algebra into the other.

\subsection{Algebraic ER=EPR in the stringy regime?} \label{sec:ER} 

In \cite{Engelhardt:2023xer}, an algebraic version of ER=EPR was proposed in the semiclassical regime to characterize spacetime connectivity with the algebras of the boundary theory. 
Here we consider the possible extension of this proposal to the stringy regime. In particular, we discuss the subtle relationship between the value of the modular depth parameter and the type of the underlying von Neumann algebra.

Consider first the $\lam \to \infty$ limit, and a two-sided pure state with a {\it classical} dual spacetime $M$. Suppose $\sX_R, \sX_L$ are respectively operator algebras of CFT$_R$ and CFT$_L$ in the large $N$ limit. Then the algebraic ER=EPR proposal says that $M$ is: 
\bea\label{i01}
&1. & \text{is disconnected iff ${\cal X}_{R}$ and ${\cal X}_{L}$ are both type I.} \\
&2. & \text{has a classical wormhole connecting $R$ to $L$ iff ${\cal X}_{R}$ and ${\cal X}_{L}$ are both type III$_{1}$.} \ \
\label{i02}
\eea
It would be desirable to generalize the proposal to the stringy regime.

In establishing~\eqref{i01}--\eqref{i02} in the algebraic ER=EPR proposal for a classical spacetime~\cite{Engelhardt:2023xer}, it was assumed that the entanglement wedges of the $R$ and $L$ systems are themselves classical spacetime manifolds. 
Such an assumption can no longer be made in the stringy regime. For example, in the case of the IOP model, we see that despite the boundary algebra being type III$_1$, the depth parameter\footnote{In that example, modular and causal depth parameters coincide.} vanishes. While the IOP model may not be dual to a string theory,\footnote{In particular it does not have the analogous $\lam \to \infty$ limit that is dual to the semi-classical gravity regime.} this example highlights that 
in general type III$_1$ by itself does not imply $\sT_m = \infty$. Additional physical inputs may be needed to specify what kind of boundary systems are dual to a string theory.  

We may also ask the converse mathematical question: does  $\mathcal{T}_m =\infty$ imply type III$_1$? 

Again, the answer is no. The reason is that there are some exotic forms of the modular spectral density that lead to an infinite depth parameter while the von Neumann algebra still has type I. In the thermofield double state at inverse temperature $\beta$, a necessary and sufficient condition for the large $N$ algebra to be type I is that the operator $\gamma_\beta$ on $L^2(\mathbb{R},\Theta(\omega)\rho(\omega)d\omega)$ defined by \begin{align}\gamma_\beta f(\omega):=e^{-\beta\omega}f(\omega)\end{align} is trace-class \cite{Derezinski,Gesteau:2024dhj}, i.e., that the support $\{\omega_i\}$ of $\rho$ is discrete and 
\begin{align}
\sum_ie^{-\beta\lvert\omega_i\rvert}<\infty.\label{eq:traceclass}
\end{align}

Theorem 3 of \cite{poltoratski2013problem} shows that it is possible to have Equation \eqref{eq:traceclass} hold true while having $\mathcal{T}=\infty$.
For an explicit example, we can again consider the ``primon gas" or ``Riemannium" at low temperature, which we saw in \ref{sec:genR} has depth parameter $\mathcal{T}=\infty$.  It has an exponentially growing spectral support but still gives rise to a type I von Neumann algebra at low enough temperature \cite{Bost1995HeckeAT}.  In this case, 
\begin{align}
\mathrm{Tr}\,{\gamma_\beta}=
\sum_ne^{-\beta\,{\mathrm{ln}\, p_n}}=\sum_np_n^{-\beta},
\end{align}
where $p_n$ denote the $n^{th}$ prime number. This series is convergent at low enough temperature, for $\beta>1$. So at low enough temperature, we are in a case in which even though the von Neumann algebra is type I, we have $\mathcal{T}=\infty$.

\subsection{Evaporating black holes: differentiating type III$_1$ algebras before and after the Page time}
\label{sec:evap}

In this subsection we give another application of the modular depth parameter by showing how a question raised in~\cite{Engelhardt:2023xer} about spacetime connectivity for an evaporating black hole system can in principle be addressed using the techniques developed in this paper. 

Consider an evaporating black hole system in AdS as discussed in \cite{Engelhardt:2023xer}. 
 Shortly before the Page time, the QES for the boundary CFT (i.e. taking $A$ to be a full boundary Cauchy slice) is empty and the black hole ($B$) is seemingly disconnected from the radiation ($R$), whereas shortly after the Page time, the QES is nontrivial and the black hole is classically connected with the radiation. See Fig.~\ref{fig:evap}.

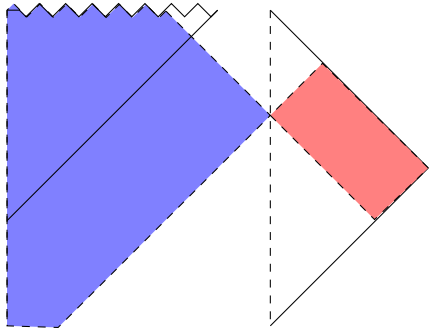
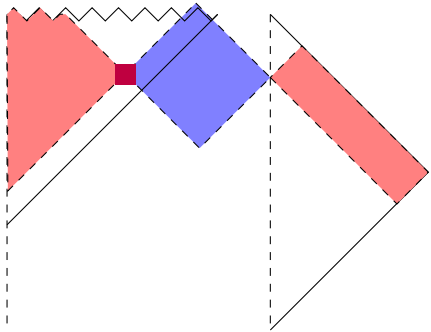
\begin{figure}
\centering
\begin{subfigure}{0.7\textwidth}
\quad\quad\quad
\begin{tikzpicture}[scale=0.7]
\node (I)    at ( 5,0)   {};
\node (Ip)    at (6,0)   {};
\node (II)   at (1,0)   {};
\path 
   (I) +(90:2)  coordinate (Itop1)
       +(-90:2) coordinate (Ibot1)
       +(180:2) coordinate (Ileft1)
       +(90:1) coordinate (Ip1)
       ;

\path 
   (Ip) +(-90:2) coordinate (Ibott)
        +(90:4) coordinate (Itopp)
        +(90:2) coordinate (C)
       ;

\path 
   (I) +(90:4)  coordinate (Itop)
       +(-90:2) coordinate (Ibot)
       +(180:4) coordinate (Ileft)
       ;

\path 
   (II) +(90:4)  coordinate (IItop)
       +(-90:2) coordinate (IIbot)
       +(180:4) coordinate (IIleft)
       +(90:0)   coordinate                  (IItop1)
       ;

\path 
(Ip1)+(0:4) coordinate (bath)
;
\path 
 (C) +(45:1.4) coordinate (Cbtop)
     +(-45:2.8) coordinate (Cbbot)   
     ;

\path 
 (C) +(45:1.4) coordinate (Cbtop)
     +(-45:2.8) coordinate (Cbbot)   
     +(135:2.8) coordinate (Cbhtop)
     +(-135:5.7) coordinate (Cbhbot)   
     ;

\draw[dashed,fill=blue!50,decoration=zigzag] (Cbhtop)--(C) -- (Cbhbot)--(IIbot)--(IItop) decorate{--(Cbhtop)}--cycle;
\draw (IItop1) -- (Itop);
\draw (Itopp) -- (bath) -- (Ibott);
\draw[decoration=zigzag] (Itop) decorate{-- (IItop)};
\draw[dashed] (IItop) -- (IIbot);
\draw[dashed] (Itopp) -- (Ibott);
\draw[dashed,fill=red!50] (C) -- (Cbbot) -- (bath) -- (Cbtop)-- (C) -- cycle;
\end{tikzpicture}
\caption{An evaporating black hole shortly before the Page time. The entanglement wedge of the bath of radiation, denoted in red, is disconnected from the entanglement wedge of the black hole denoted in blue: they do not share any nontrivial quantum extremal surface.}
\end{subfigure}
\begin{subfigure}{0.7\textwidth}
\quad\quad\quad
\begin{tikzpicture}[scale=0.7]
\node (I)    at ( 5,0)   {};
\node (Ip)    at (6,0)   {};
\node (II)   at (1,0)   {};
\path 
   (I) +(90:2)  coordinate (Itop1)
       +(-90:2) coordinate (Ibot1)
       +(180:2) coordinate (Ileft1)
       +(90:1) coordinate (Ip1)
       ;

\path 
   (Ip) +(-90:2) coordinate (Ibott)
        +(90:4) coordinate (Itopp)
        +(90:2.8) coordinate (C)
       ;

\path 
   (I) +(90:4)  coordinate (Itop)
       +(-90:2) coordinate (Ibot)
       +(180:4) coordinate (Ileft)
       ;

\path 
   (II) +(90:4)  coordinate (IItop)
       +(-90:2) coordinate (IIbot)
       +(180:4) coordinate (IIleft)
       +(90:0)   coordinate                  (IItop1)
       ;

\path 
(Ip1)+(0:4) coordinate (bath)
;
\path 
 (C) +(45:1.4) coordinate (Cbtop)
     +(-45:2.8) coordinate (Cbbot)   
     ;

\path 
 (C) +(45:0.85) coordinate (Cbtop)
     +(-45:3.4) coordinate (Cbbot)   
     +(135:2) coordinate (Cbhtop)
     ;
\path
 (Cbhtop) +(-135:1.9) coordinate (Cbhbot)
 ;
\path 
(Cbhbot) +(-45:2) coordinate (radb)
;
\path 
(Cbhbot) +(-135:3.15) coordinate (bleft)
       +(135:1.6) coordinate (btop)
;

\draw[dashed,fill=blue!50] (Cbhtop)--(C) -- node[midway, below, sloped] {}(radb)--(Cbhbot)--(Cbhtop)--cycle;
\draw (IItop1) -- (Itop);
\draw (Itopp) -- (bath) -- (Ibott);
\draw[decoration=zigzag] (Itop) decorate{-- (IItop)};
\draw[dashed] (IItop) -- (IIbot);
\draw[dashed] (Itopp) -- (Ibott);
\draw[dashed,fill=red!50] (C) -- (Cbbot) -- (bath) -- (Cbtop)-- (C) -- cycle;
\draw[dashed,fill=red!50,decoration=zigzag] (Cbhbot)--(bleft)--(IItop) decorate{-- (btop)}--(Cbhbot);
\node[fill=purple] (QES) at (3.25,2.85) {};
\end{tikzpicture}
\caption{An evaporating black hole shortly after the Page time. The entanglement wedge of the bath of radiation, denoted in red, is now connected to the entanglement wedge of the black hole denoted in blue: they do share a quantum extremal surface represented by the purple dot.}
\end{subfigure}
\caption{An evaporating black hole shortly before and shortly after the Page time. In both cases, the algebra of observables describing the interior has type III$_1$, however it is only after the Page time that a QES creates connectivity between the entanglement wedge of the radiation and the entanglement wedge of the black hole.}
\label{fig:evap}
\end{figure}

On both sides of the Page time, the corresponding boundary algebra $\sX_B$ for the black hole subsystem has been argued to be type III$_1$. 
Thus it appears to provide a counterexample to~\eqref{i02}. 

It is not a counterexample, as in this case the dual bulk spacetime is not classical, but quantum volatile\footnote{In a quantum volatile spacetime, one or both of the following are violated in the $G_N \to 0$ limit: (i)
 fluctuations of diffeomorphism invariants go to zero as ${\cal O}(G_{N}^{a})$ for some $a>0$;  (ii) spacetime diffeomorphism invariants such as volumes, areas, and lengths  do not scale as $O(G_{N}^{-a})$ with $a > 0$.}---the volume of the black hole interior grows to infinity in the $G_N \to 0$ limit. It was accordingly argued in~\cite{Engelhardt:2023xer} that, before the Page time, $B$ and $R$ are quantum connected. However, it was not clear how to distinguish directly the classical connectivity after the Page time from the quantum connectivity before the Page time intrinsically from the structure of the boundary algebra $\sX_B$. 
 
We propose that the modular depth parameter can diagnose this change of connectivity, with $\mathcal{T}_m<\infty$ before the Page time, and $\mathcal{T}_m=\infty$ after the Page time. That $\mathcal{T}_m=\infty$ after the Page time follows simply from existence of a non-empty QES. We now outline an argument for 
$\mathcal{T}_m<\infty$ before the Page time. Since currently we cannot calculate the modular spectral function for $\sX_B$ explicitly, the argument we present below should be only viewed as a plausibility argument.

Before the Page time, the entanglement wedge of the boundary includes the interior of the black hole which has an infinite proper length in the $G_N \to 0$ limit. If it were possible to probe this region using light rays, we would have concluded that the depth parameter should be infinite. We now argue that in terms of the modular depth parameter, it can in fact be finite. 
 
To illustrate the main idea, we will use a simple model. Consider a $(1+1)$-dimensional (nonrelativistic) free field theory of a scalar field $\Phi$ defined on an infinite half line with boundary at $x=0$. Consider a localized excitation of the bulk, say by a bump function $f_0$ localized on a small interval of a one-dimensional time slice $t=0$ (at infinite volume), of the form \be 
\Phi (f) = \int_{-\infty}^0 dx \, f (x) \Phi (x), \quad f \in L^2(\mathbb{R_-}) \ .
\ee Importantly, we allow $f$ to be complex-valued, so that $\phi(f)$ is in general a linear combination of two selfadjoint fields that can be interpreted as position and momentum. We have the commutation relation \begin{align}[\phi(f),\phi(g)]=i\,\mathrm{Im}\braket{f,g}Id.\end{align}As a toy model, suppose that the bulk modular Hamiltonian can be modelled on the one particle Hilbert space $L^2(\mathbb{R}_-)$ by some selfadjoint operator $H$. Then in the second quantized picture, the operator $\phi(f)$ is evolved at time $t$ into the operator $\phi(f_t)$, where $f_t(x)$ is the value of the solution $f(x,t)$ to the Schr\"odinger equation \begin{align}i\partial_t f=H f,\end{align} with initial condition \begin{align}f(x,0)=f_0(x).\end{align}
Out of this data, we can construct a generalized free field theory in $0+1$ dimensions, as follows. First, define an antisymmetric function of time 
\begin{align}
\rho(t):=\braket{[\phi(f_t),\phi(f_0)]}_\beta=i\,\mathrm{Im}\braket{f_t,f_0}.
\end{align}
Then construct a generalized free field theory at finite temperature in the same way as in Section \ref{sec:spec}, from this spectral function. We want to show that for a physically reasonable choice of Hamiltonian $H$, $\mathcal{T}<\infty$.

As a toy example, we can choose $H$ to be the Hamiltonian of a one-dimensional quantum harmonic oscillator. Since the odd eigenfunctions of the harmonic oscillator form a basis of $L^2(\mathbb{R}_-)$, for any choice of $f_0\in L^2(\mathbb{R}_-)$, we can decompose \begin{align}f_0=\sum_{n\in2\mathbb{N}+1}\lambda_n\psi_n,\end{align}where the $\psi_n$ are the harmonic oscillator eigenfunctions of energy $(n+\frac{1}{2})\omega_0$. As a result, $\rho(t)$ is of the form
\begin{align}
\rho(t)=\sum_{n\in\mathbb{Z}} \mu_n e^{i\left(n+\frac{1}{2}\right)t},
\end{align}
which means that in frequency space it is an evenly spaced linear combination of Dirac deltas. Therefore it is periodic up to a phase, and we deduce that the depth parameter $\mathcal{T}$ for this generalized free field is finite (and smaller than $2\pi$).

While modular flows are nonlocal, it may still sound surprising that it can give a finite modular depth for an infinite proper length.

Heuristically, we do not expect the structure of modular flow outside the black hole to significantly differ from the one of the thermofield double state. In the exterior region, we expect that operators can be reconstructed from ``simple" single trace operators on the boundary with good accuracy, and in particular that at least for reasonably small time scales, the modular flows of these operators can be approximated by boundary time evolution of these single trace operators. 
But modular flow on operators in the interior significantly differs from a thermofield double. These operators are complex and cannot be reconstructed solely from boundary single traces: we also need large $N$ limits of boundary modular flows. On these operators, modular flow can potentially act very differently from boundary time evolution.
Hence, outside the black hole we expect operators to propagate under modular flow following (approximately) relativistic laws (in particular, with finite speed), but close enough to the interior, modular flow is so different from any geometric transformation that it could in principle allow for a modular Hamiltonian that is nonlocal enough that it guarantees that the value of $\mathcal{T}$ is finite.

\section{Discussion} \label{sec:diss}

In this paper, we have taken a first step towards developing a new language to describe
causal structure, event horizons, and quantum extremal surfaces (QES) for the bulk duals of boundary systems that extend beyond the standard Einstein gravity regime. Here we comment on some future directions.

\ben 

\item {\bf Connections between geometric concepts and the ergodic hierarchy/quantum chaos}

In this paper we introduced the depth parameter $\mathcal{T}$ by considering the inclusion structure of time band algebras on the boundary. We commented earlier that the depth parameter has an interpretation in terms of ergodicity,
 in the sense that it encodes how long one needs to wait in order to obtain full information about the large $N$ theory. 
 
We also defined horizons in terms of half-sided inclusion, which is related to a stronger notion of ergodic or quantum chaotic behavior. 

As recently discussed in~\cite{Gesteau:2023rrx,Ouseph:2023juq}, there are various other possible levels of ergodicity and chaos in a system, beyond those corresponding to the depth parameter and half-sided modular inclusion. It 
would be interesting to understand whether they can be associated with emergent bulk geometric structure. 
Conversely, bulk geometric structures could give inspirations for formulating a mathematical theory of quantum chaos, which is still at early stages. 

\item  {\bf Structure of higher-point functions}

The structures discussed in this paper can be deduced with the knowledge of boundary two-point functions. 
Exploring the structure of higher-point functions could lead to new insights into the non-locality brought by possible violation of equivalence principle mentioned earlier. 

Furthermore, at the level of Einstein gravity, the presence of a black hole event horizon leads to universal behavior 
in out-of-time-ordered correlators (OTOCs)~\cite{SheSta14}, and maximal value of the quantum Lyapunov exponent~\cite{Maldacena:2015waa}. 
Going to the stringy regime, OTOCs involve exchanges of an infinite tower of higher-spin stringy excitations, which gives rise to non-maximal chaos~\cite{SheSta15}. It would be interesting to understand whether such behavior can still be attributed to a stringy horizon.

\item {\bf Connection to the inverse spectral problem}

The question we are asking here is closely related to the so-called inverse spectral problem. The basic idea is the following. In an asymptotically AdS spacetime, the equation of motion for a free field reduces to a Schrödinger equation in the radial variable, with some potential $V(r)$ dictated by the geometry. The standard formulation of the inverse scattering problem is to ask, given the scattering data for the solutions of the Schrödinger equation, how one may reconstruct the potential $V(r)$. 

In our context, the scattering data corresponds to the correlation functions of the boundary generalized free fields, while the potential $V(r)$ would correspond to a ``stringy geometry" in the bulk. There is a general approach to the inverse spectral problem, due to Gelfand and Levitan \cite{Gelfand_Levitan_1955a}, which allows to recover quite explicit results on $V(r)$ given the scattering data. Interestingly, Gelfand--Levitan theory is intimately related to the exponential type problem in harmonic analysis, which we saw that in our context, maps to the determination of the depth parameter $\mathcal{T}$. Therefore it is likely that the ideas introduced here could find a very natural place within Gelfand--Levitan theory, and that understanding how these ideas fit together could teach us more about emergent stringy geometries, especially about modular Hamiltonians of time bands, quasi-normal modes of stringy black holes, and maybe even a description of stringy singularities. We hope to report on this connection in the near future.

\item {\bf Towards a stringy QES formula}

In this work, we were able to construct an algebra of sequences to describe stringy bifurcate horizons. What we are still lacking is some notion of measure on this algebra that could give us information on a stringy notion of ``area" of the bifurcate horizon. Coming up with such a notion would open a path towards a possible derivation of a Ryu--Takayanagi/QES formula at nonzero string length. 

The operator algebraic language has already proven to be useful in order to understand gravitational entropy, in particular the works \cite{Witten:2021unn,Chandrasekaran:2022eqq,Chandrasekaran:2022cip,Jensen:2023yxy,Kudler-Flam:2023qfl} on the crossed product have allowed to define a type II entropy that agrees with the semiclassical notion of generalized entropy up to an overall constant. However, this semiclassical notion does not provide any meaningful way of constructing an area term, as 
the type II entropy is only defined up to a constant.  So far, the only meaningful notion of area term in the operator algebraic context have been formulated by mapping the large $N$ theory to the finite $N$ theory through an approximate quantum error correcting code \cite{Faulkner:2022ada,Gesteau:2023hbq}. It is our hope that the methods introduced here may be relevant to define a notion of area term directly at large $N$, for arbitrary values of the string length, that contributes to gravitational entropy together with the entropy of bulk fields. We expect that such a setup could also provide new insight into the exchange ambiguity between the bulk entropy term and the area term \cite{Susskind:1994sm,Gesteau:2023hbq}.

\item {\bf A new language for stringy spacetime?}

Our results suggest that the algebraic framework is capable of retaining properties that one expects from an emergent spacetime even in a potentially very stringy regime where this spacetime is not described by the usual notion of smooth differentiable Lorentzian manifold. In particular, it appears that the algebraic structure allows to describe the causal structure of this stringy spacetime through the observables that live on some of its subregions.

The idea of describing geometric spaces through the objects that live on their subregions has been extremely fruitful in mathematics, and is at the core of most of modern geometry. In topology, this is exemplified by Gelfand duality, which describes a space by its algebra of functions, while in algebraic geometry this is for example achieved by the notion of sheaf (and actually one can see many of the constructions in algebraic quantum field theory through this lens). It would be very interesting to see how far this idea can be pushed in our context, and find whether one can \textit{define} the subregions of a stringy spacetime only through the algebras of observables that live on its subregions and their embeddings. This way of defining a space directly through the algebras that characterize its subspaces is somewhat reminiscent of the notion of topos, due to Grothendieck. A hint that such an approach to stringy geometry may be possible is that von Neumann algebras and their commutants on a Fock space already come equipped with a structure: that of a complemented complete lattice \cite{Derezinski}, which is reviewed in Appendix \ref{app:araki}. In particular we can construct a set of real symplectic subspaces $\mathcal{V}_i$ of the large $N$ one particle Hilbert space, stable generation, intersection and symplectic complement, and associated von Neumann algebras $\mathcal{M}(\mathcal{V}_i)$, such that
\begin{thm}
\bega
\mathcal{M}(\mathcal{V}_1)=\mathcal{M}(\mathcal{V}_2)\;\mathrm{iff}\;\mathcal{V}_1=\mathcal{V}_2, \\
\mathcal{V}_1\subset\mathcal{V}_2\Rightarrow\mathcal{M}(\mathcal{V}_1)\subset\mathcal{M}(\mathcal{V}_2), \\
\mathcal{M}\left(\bigvee_i \mathcal{V}_i\right)=\bigvee_i \mathcal{M}(\mathcal{V}_i), \\
\mathcal{M}\left(\bigcap_i \mathcal{V}_i\right)=\bigcap_i \mathcal{M}(\mathcal{V}_i), \\
\mathcal{M}\left(\mathcal{V}\right)^\prime=\mathcal{M}\left(\mathcal{V}^\prime\right),
\end{gather}
and $\mathcal{M}\left(\mathcal{V}\right)$ is a factor iff $\mathcal{V}\cap\mathcal{V}^\prime=\{0\}.$
\end{thm}

This structure is ideal in order to retain a notion of causality, since the commutant operation (or the symplectic complement at the one particle level) can be understood as an operation mapping a region to its causal complement. What now remains to be seen is how much geometric structure can be put on top of this causal structure in the stringy regime - how close to a full-fledged Lorentzian manifold can we get?

\een

\vspace{0.2in}   \centerline{\bf{Acknowledgements}} \vspace{0.2in}
We would like to thank Suzanne Bintanja, Matthew Dodelson, Netta Engelhardt, Tom Faulkner, Matthias  Gaberdiel, Elias Kiritsis, Sam Leutheusser, Nikolai Makarov, Juan Maldacena, Matilde Marcolli, Pascal Millet, Vladimir Narovlansky, Alexei Poltoratski, Leonardo Santilli, Joaquin Turiaci, and Edward Witten for discussions. 
This work is supported by the Office of High Energy Physics of U.S. Department of Energy under grant Contract Number  DE-SC0012567 and DE-SC0020360 (MIT contract \# 578218).

\appendix
\section{Elements of distribution theory}
\label{app:distributions}
In this first appendix, we recall the basic points of distribution theory that are being used throughout the paper. We particularly emphasize how one can take the Fourier transform of a distribution, as Fourier theory is the central element of our discussion.

\subsection{Distributions and tempered distributions}

Informally, distributions are a generalized notion of function that only makes sense \textit{dually}. Consider the example of the real line. A natural space of functions on $\mathbb{R}$ would be the space of square-integrable functions $L^2(\mathbb{R})$. However, in some contexts one would like to be able to consider more general kinds of ``functions" on $\mathbb{R}$, for example, discrete measures such as Dirac's delta, or oscillatory functions like $e^{i\omega t}$.

The issue is that such oscillatory functions, if defined at all, are definitely not square-integrable. This issue is often encountered in quantum mechanics, for example when considering a free particle in one dimension. One would like to diagonalize the Hamiltonian in a ``basis of Dirac deltas", but they are not part of the Hilbert space of $L^2(\mathbb{R})$. The well-known solution to this problem, which is sometimes referred to as ``rigging" the Hilbert space, amounts to allow more general notions of wavefunctions by passing to a space of distributions.

The idea is the following. By the Riesz representation theorem, the Hilbert space $\mathcal{H}=L^2(\mathbb{R})$ is \textit{reflexive}, which means that it is isomorphic to its topological dual. Therefore there are exactly as many continuous linear functionals on $\mathcal{H}=L^2(\mathbb{R})$ as there are elements of $\mathcal{H}=L^2(\mathbb{R})$. Therefore, if one wants to find a space of generalized functions that is larger than $L^2(\mathbb{R})$, a natural thing to do is to look at a \textit{subspace} of $L^2(\mathbb{R})$, so that there are more continuous linear functionals on it than on $L^2(\mathbb{R})$. 

The choice of such a subspace is nonunique, and leads to different notions of distributions. The most obvious choice of dense subspace of $L^2(\mathbb{R})$ is the space of compactly supported test functions $\mathcal{D}(\mathbb{R})$. It turns out that this space can be equipped with a topology, called the \textit{inductive limit topology}, such that for this topology, the topological dual $\mathcal{D}^\prime(\mathbb{R})$ of $\mathcal{D}(\mathbb{R})$ is strictly larger than $\mathcal{D}(\mathbb{R})$. In particular it contains Dirac deltas, and a large family of continuous functions. The space $\mathcal{D}^\prime(\mathbb{R})$ is known as the space of \textit{distributions} on $\mathbb{R}$.

Another interesting choice of dense subspace of $L^2(\mathbb{R})$ is the so-called \textit{Schwartz space} $\mathcal{S}(\mathbb{R})$, of functions whose derivatives all decay faster than polynomially at infinity. There also exists a way to equip $\mathcal{S}(\mathbb{R})$ with a natural topology for which $L^2(\mathbb{R})\subset\mathcal{S}^\prime(\mathbb{R})$. In fact, we even have a stronger statement: as \begin{align}\mathcal{D}(\mathbb{R})\subset\mathcal{S}(\mathbb{R})\subset L^2(\mathbb{R}),\end{align}
we have 
\begin{align}L^2(\mathbb{R})\subset\mathcal{S}^\prime(\mathbb{R})\subset \mathcal{D}^\prime(\mathbb{R}).\end{align}

Therefore one can see the notion of tempered distribution as a refined notion of distribution, that is continuous on a larger space of test functions (all of Schwartz space). One of the main reasons why tempered distributions are useful is because they are the right kind of distribution to consider in order to take Fourier transforms, as we now review.

\subsection{Fourier transforms of distributions}

The classical formula for the Fourier transform, 
\begin{align}
\tilde{f}(\omega):=\int_{\mathbb{R}}f(t)e^{-i\omega t} dt,
\end{align} is only well-defined on $L^1(\mathbb{R})$. However tempered distributions can be much more general. For example any continuous function with subexponential growth defines a tempered distribution. The point, however, is that the dual picture allows to define a notion of Fourier transform for any tempered distribution. The reason why we need the distributions to be tempered is because we need the space of test functions we choose to be stable under Fourier transform. As we point out now, this is obviously not true for compactly supported functions, but it is true for Schwartz space:

\begin{prop}
The Schwartz space $\mathcal{S}(\mathbb{R})$ is stable under Fourier transform.
\end{prop}

This result allows to write a definition for the Fourier transform of a general tempered distribution.

\begin{defn}
Let $T$ be a tempered distribution on $\mathbb{R}$. We define the Fourier transform $\tilde{T}$ of the tempered distribution $T$ as follows. For $f\in\mathcal{S}(\mathbb{R})$,
\begin{align}
\tilde{T}(f):=T(\tilde{f}).
\end{align}
\end{defn}

In general, the Fourier transform of a tempered distribution is a general tempered distribution, however, under some extra assumption on the tempered distribution, one can say additional things about the regularity of the Fourier transform. There are many results along this line, going from very easy ones to highly nontrivial results in harmonic analysis like the Beurling--Malliavin theorem or solutions to the type problem we will discuss in the next appendix. Here we give one of the most useful ``easier" results of this type, namely an important part of the Schwartz--Paley--Wiener theorem:

\begin{prop}
The Fourier transform of a compactly supported distribution extends to an entire function on the complex plane.
\end{prop}

\subsection{Pointwise product and convolution of distributions}

Most of the time in this work, we reason in Fourier space on tempered distributions whose Fourier transforms are square integrable functions against a suitable measure, hence, it is usually not a problem to pointwise multiply the Fourier transforms. Nevertheless, it is worth pointing out that in general there is an obstruction to performing a pointwise product of distributions. For example, two Dirac deltas cannot be pointwise multiplied because they are all singular at the same point. It is therefore useful to ask when, in general, it is possible to define the pointwise product of distributions. It turns out that the answer relies on the notion of wavefront set, which we now define.

\begin{defn}
Let $\Omega$ be an open set in $\mathbb{R}^n$. Given a distribution $T\in\mathcal{D}^\prime(\mathbb{R}^n)$, the wavefront set $WF(T)$ of $T$ is defined to be the complement of the union of all conic subsets of $T^\ast U$ on which $T$ is smooth.
\end{defn}
This definition extends straightforwardly to the case of manifolds. We then have the crucial property:
\begin{prop}
The pointwise product of two distributions $T_1$ and $T_2$ is well-defined unless there exists $(x,k)\in T^\ast U$ such that $(x,k)\in WF(T_1)$ and $(x,-k)\in WF(T_2)$. 
\end{prop}

Similarly, since the convolution of two distributions requires mutiplying them pointwise, the notion of wavefront set is also relevant in order to define the convolution product of two general distributions.

\section{Generalized free field theory and review of  a theorem of Araki} \label{app:araki} 

In this Appendix, we review the construction of the Hilbert space associated with a generalized free field theory 
around a semi-classical state, and then discuss a theorem of Araki~\cite{Araki_1963} that is used in the main text.

\subsection{Construction of a generalized free field theory} 

Consider a single-trace Hermitian operator $\phi (t)$ (again we suppress spatial dependence for notational simplicity), with 
spectral function 
\be \label{s1}
\rho (t-t') = \vev{\Psi|[\phi (t) ,\phi (t')]|\Psi}  ,
\ee
in a semi-classical state $\ket{\Psi}$. For simplicity we will assume that $\ket{\Psi}$ is time translation invariant. 
$\rho(t)$ and its Fourier transform $\rho (\om)$ have the properties 
\be 
\rho^*(t) = \rho (-t) = - \rho (t), \quad \rho^* (\om) =  \rho (\om) =- \rho (-\om), 
\quad \th (\om) \rho (\om) \geq 0  \ . 
\ee
Denote the Fourier transform of $\phi (t)$ as $\phi_\om$, we have from~\eqref{s1} 
the commutation relations  
\be
 [\phi_\om, \phi_{\om'}] = \rho (\om) 2 \pi \de (\om + \om') , \quad
\phi_{-\om} = \phi_\om^\da\ .
\ee

Introduce smeared operators 
\be 
\phi (f) = \int d t \, f (t) \phi (t) \equiv  \int {d \om \ov 2 \pi} f(-\om) \phi_\om, \quad f (-\om) = f^* (\om), 
\ee
where $f(t)$ is real such that $\phi (f)$ is Hermitian. 
 $\rho(t)$ can be used to define a symplectic product (i.e. anti-symmetric between $f (t)$ and $g (t)$) in the space $\{f (t)\}$ 
\be \label{sym}
(f,g) = - i \vev{\Psi|[\phi (f) ,\phi (g)]|\Psi} = - i \int dt dt' \, f(t) \rho (t-t') g (t') 
= - i \int {d \om \ov 2 \pi} \, f (-\om) \rho (\om) g (\om) \ .
\ee
We can also associate a complex vector $\ket{f}$ with $(f (\om), \om > 0)$ and introduce an inner product 
\be 
\vev{f,g} = 2 \int {d\om \ov 2 \pi} \, f^* (\om) \, \th (\om) \rho (\om) \, g (\om) \ .
\ee
Upon completion, the set of normalizable vectors then form a Hilbert space $\sZ$, which is often referred to as the single-particle Hilbert space. Depending on the state, we will need to restrict admissible test functions to the ones that sit in the domain of a certain operator to be defined below. The space $\sV$ of $f$ (i.e. the space of $f(\om), \forall \om$) is a closed real subspace of $\sZ$, and  
the symplectic product~\eqref{sym} is defined on this real subspace, with 
\be 
(f,g) = {\rm Im} \vev{f,g} \ .
\ee

We now consider several explicit examples. 

\subsubsection{Vacuum state} 

For $\ket{\Psi}$ to be vacuum state $\ket{\Om}$, $\phi_\om$ with $\om > 0$ is taken to be the annihilation operators, i.e. 
\be 
 \phi_\om \ket{\Om} =0 , \quad \om > 0 \ .
\ee 
The Fock space, usually denoted as $\Ga (\sZ)$,  can then be obtained acting creation operators $\phi_\om , \om < 0$ on the vacuum.

\subsubsection{TFD state} 

For $\ket{\Psi}$ given by the TFD state, we take $\phi$ to be the operator $\phi^R$ in the CFT$_R$. There is also a corresponding $\phi^L$ with the corresponding single-particle Hilbert space denoted as $\bar \sZ$. In this case no linear combinations of $\phi^R$ can annihilate the TFD state. Instead, we have 
\be \label{hheq}
\phi_{\om}^{(R)} \ket{\rm TFD} = e^{-{\b \om \ov 2} }  \phi_{-\om}^{(L)} \ket{\rm TFD}  \ .
\ee
We can then introduce 
\be\label{c111HH}
c^{(R)}_\om = \fb_+ \phi^{(R)}_\om - \fb_- \phi^{(L)}_{-\om}, 
\quad c^{(L)}_\om = \fb_+ \phi^{(L)}_\om - \fb_- \phi^{(R)}_{-\om}, 
\quad 
 \fb_\pm \equiv \le( {e^{\pm {\b|\om| \ov 2}} \ov 2 \sinh {\b |\om| \ov 2}}\ri)^\ha,
\ee
with 
\be 
 c^{(R)}_\om \ket{\rm TFD} = c^{(L)}_\om \ket{\rm TFD} =0 , \quad \om > 0 \ .
\ee 
The Fock space $\Ga (\sZ \oplus \bar \sZ)$ can be generated by acting $c^{(R,L)}_\om, \om < 0$ on $\ket{\rm TFD}$.

In order to define creation operators associated to a function in $\mathcal{Z}$, we also must require $f$ to be in the domain of the operator of multiplication by $\fb_-$.

\subsubsection{A quasi-free state} 

A quasi-free state $\ket{\Psi_\ga}$ is a two-sided state defined as 
\be 
\phi_{\om}^{(R)} \ket{\Psi_\ga} = \ga^{\ha}  (\om)  \phi_{-\om}^{(L)} \ket{\Psi_\ga}  \ .
\ee
We can then introduce 
\bega\label{c11HH}
c^{(R)}_\om = \fb_+ \phi^{(R)}_\om - \fb_- \phi^{(L)}_{-\om}, 
\qquad c^{(L)}_\om = \fb_+ \phi^{(L)}_\om - \fb_- \phi^{(R)}_{-\om}, 
\\
 \fb_+ \equiv \le( {\ga^{-\ha}  \ov \ga^{-\ha}   - \ga^{\ha}  }\ri)^\ha = (1 + \hat \rho_\ga)^\ha,\quad
  \fb_- \equiv \le( {\ga^{\ha} (\om) \ov \ga^{-\ha} (\om)  - \ga^{\ha} (\om) }\ri)^\ha = \hat \rho_\ga^\ha, \quad 
 \hat \rho_\ga \equiv {\ga \ov 1-\ga}  ,
\end{gather} 
with 
\be 
 c^{(R)}_\om \ket{\Psi_\ga} = c^{(L)}_\om \ket{\Psi_\ga} =0 , \quad \om > 0 \ .
\ee 
The Fock space $\Ga (\sZ \oplus \bar \sZ)$  can be generated by acting $c^{(R,L)}_\om, \om < 0$ on $\ket{\Psi_\ga}$. 

In order to define creation operators associated to a function in $\mathcal{Z}$, we also must require $f$ to be in the domain of the operator of multiplication by $\fb_-$.

\subsection{Symplectic complement and commutant}

Consider a single particle Hilbert space $\sW$ and the corresponding Fock space $\Ga (\sW)$. 
For the examples discussed earlier, for the vacuum we have $\sW = \sZ$, while for the TFD and a general quasi-free state we have $\sW = \sZ \oplus \bar \sZ$. 

Introduce Weyl operators 
\be 
W (f) \equiv e^{i \phi (f)} , \quad f \in \sW 
\ee
and define a von Neumann algebra 
\begin{align}
\mathcal{M}(\mathcal{W}) \equiv \{W(f),f\in\mathcal{W}\}^{\prime\prime}.
\end{align}
It can be shown that
\be 
\sM (\sW) = \sB (\Ga (\sW))  \ .
\ee
That is, the algebra of bounded operators 
on $\Ga (\sW)$ is the Weyl algebra of operators on the single-particle Hilbert space.

Now suppose $\mathcal{V}$ is closed real subspace of $\mathcal{W}$. 

\begin{defn}We define the \textit{symplectic complement} of $\mathcal{V}$ as
\begin{align}
\mathcal{V}^\prime=\{z\in\mathcal{W},\,\forall v\in\mathcal{V}, \mathrm{Im}\braket{z,v}=0.\}
\end{align}
\end{defn}
\begin{defn} We define the von Neumann algebra $\sM (\sV)$ associated with $\sV$ as 
\begin{align}
\mathcal{M}(\mathcal{\sV}) \equiv \{W(v),v \in \sV \}^{\prime\prime}.
\end{align}
\end{defn}
The following theorem due to Araki~\cite{Araki_1963} says that the commutant operation of the von Neumann algebra and 
 the symplectic complement operation are mapped to each other.
\begin{thm}
\bega
\mathcal{M}(\mathcal{V}_1)=\mathcal{M}(\mathcal{V}_2)\;\mathrm{iff}\;\mathcal{V}_1=\mathcal{V}_2, \\
\mathcal{V}_1\subset\mathcal{V}_2\Rightarrow\mathcal{M}(\mathcal{V}_1)\subset\mathcal{M}(\mathcal{V}_2), \\
\mathcal{M}\left(\bigvee_i \mathcal{V}_i\right)=\bigvee_i \mathcal{M}(\mathcal{V}_i), \\
\mathcal{M}\left(\bigcap_i \mathcal{V}_i\right)=\bigcap_i \mathcal{M}(\mathcal{V}_i), \\
\mathcal{M}\left(\mathcal{V}\right)^\prime=\mathcal{M}\left(\mathcal{V}^\prime\right),
\end{gather}
and $\mathcal{M}\left(\mathcal{V}\right)$ is a factor iff $\mathcal{V}\cap\mathcal{V}^\prime=\{0\}.$
\label{thm:araki}
\end{thm}

This theorem is very useful, as it allows to translate algebraic statements on the full Fock space to 
to statements about the one particle Hilbert space.

\subsection{No nontrivial singular inclusions in a generalized free field theory} \label{app:equ}

Here we give a proof that for a generalized free field, if $\mathcal{S}_{(-t_0,t_0)}$ is a proper subalgebra of $\mathcal{S}_R$, then it must have a nonempty relative commutant. If we call a proper inclusion $\mathcal{S}_{(-t_0,t_0)}$ in $\mathcal{S}_R$ with empty relative commutant as a singular inclusion, then for a generalized free field, singular inclusions do not exist. This result seems standard in the algebraic QFT literature, in which significant effort has been made to construct singular inclusions that do \textit{not} come from generalized free fields, see for example \cite{correa2023modular,longo2017freeproductsaqft}. It is also a direct consequence of Theorem \ref{thm:araki}.

\begin{prop}
    Let $V_1\subset V_2$ be two closed real subspaces like in Theorem \ref{thm:araki}, and suppose that the relative commutant of $M(V_1)$ inside $M(V_2)$ is trivial (in particular, this implies $M(V_2)$ is a factor). Then, \begin{align}M(V_1)=M(V_2).\end{align}
\end{prop}

\begin{proof}
\begin{align}
M(V_2)=&M(V_2) \cap \mathcal{B}(\mathcal{H})\\
= &M(V_2)\cap (M(V_1) \vee M(V_1)')\\
=&M(V_2)\cap M(V_1 \vee (V_1)')\\
=&M(V_2 \cap(V_1 \vee (V_1)'))\\
=&M((V_2\cap V_1) \vee(V_2\cap (V_1)'))\\
=&M(V_2\cap V_1)\\
=&M(V_1).
\end{align}
So the inclusion is actually trivial.
\end{proof}
\section{A review of the exponential type problem}\label{app:exp}

We saw in the main text that for a generalized free field theory at finite temperature, the value of the depth parameter was equal to the exponential type of the Kallen-Lehmann measure. In general, it is difficult to compute the exponential type of a measure, and there is a large body of mathematical work on this topic. The goal of this appendix is to summarize this mathematical work. If one follows the general intuition given by Heisenberg's uncertainty principle, one expects that the more porous the support of the spectral density, the least porous the support of the Fourier transform of any function supported there can be. In particular, for porous enough support of the spectral density, we expect there to be a maximal gap for any distribution of the form $f\ast E$ with $f$ admissible.

Therefore, studying the triviality of large enough relative commutants of time bands is closely related to the following analysis problem, which as we saw before, is known as the \textit{exponential type problem}. 

\subsection{A review of the main results}

The mathematical results on the exponential type problem (and a close relative known as the gap problem), seem to culminate with the work of Poltoratski in 2012 \cite{poltoratski2013problem}, which gives an explicit answer to the gap problem for finite measures on the real line. Earlier work includes important theorems due to Benedicks \cite{BENEDICKS1985180,f4312568-16eb-3ae7-9d52-073fcbb92f76}, Beurling and de Branges \cite{66cbfd0f-454e-3156-b6ff-7e5822b3a818}. 

We start by stating the definition of the exponential type:
\begin{defn}
Let $\mu$ be a measure on $\mathbb{R}$. We define the \textit{exponential type} of $\mu$ as \begin{align}T_\mu:=\underset{t\in\mathbb{R}}{\mathrm{sup}}\{t\in\mathbb{R},\,\exists f\in L^2(\mu),\,\mathrm{supp}(\widetilde{f\mu})\cap(0,t)=\emptyset \}.\end{align}
\end{defn}
The exponential type problem consists in computing the value of $T_\mu$ for a measure $\mu$.

The type problem has been solved in various cases. In the examples given in the main text, we showed that \begin{itemize}\item If $\mu$ is evenly spaced with spacing $l$ then $T_\mu\leq 2\pi/l$.
\item If $\mu$ is compactly supported, then $T_\mu=0$.
\item If $\mu$ is a continuous function with complete support (with a few extra technical assumptions), then $T_\mu=\infty$.
\end{itemize}

The most general solutions to the exponential type problem to date have been provided by Poltoratski in \cite{poltoratski2013problem}, where he introduced. The work of \cite{poltoratski2013problem} is formulated in the case where the measure $\mu$ has \textit{finite} mass, or in the case of Poisson-summable measures, i.e. measures for which \begin{align}\int_{\mathbb{R}}\frac{d\lvert\mu\rvert(x)}{1+x^2}dx<\infty.\end{align}However it can be straightforwardly extended to the polynomially growing case \cite{private}.

The key ingredient of these results is the notion of $d$-uniform sequence. This notion allows to capture ``how concentrated" the support of a sequence is. The definition is a bit technical, we recall it here:

\begin{defn}
We say that the discrete sequence of real numbers $\Lambda=(\lambda_n)$ is $d$\textit{-uniform} if there exists a sequence of disjoint intervals $I_n=(a_n,b_n]$, whose centers go to $\pm\infty$ and widths go to $\infty$ as $n\rightarrow\pm\infty$, such that:
\begin{align}
\sum_{n}\frac{\lvert I_n\rvert^2}{1+\mathrm{dist}(0,I_n)^2}<\infty,
\label{eq:short}
\end{align}
\begin{align}
\#(\Lambda\cap I_n)=d\lvert I_n\rvert+o(\lvert I_n\rvert),
\label{eq:density}
\end{align}
and
\begin{align}
\sum_n \frac{(\#(\Lambda\cap I_n))^2\mathrm{log}\lvert I_n\rvert-E_n}{1+\mathrm{dist}(0,I_n)^2}<\infty,
\label{eq:energy}
\end{align}
where
\begin{align}
E_n:=\sum_{\lambda_k\neq\lambda_l\in I_n}\mathrm{log}\lvert\lambda_k-\lambda_l\rvert.
\end{align}
\end{defn}

Roughly speaking, Equation \eqref{eq:short} means that the sequence of sampling intervals $I_n$ cannot cover too much of the real line. Equation \eqref{eq:density} controls the density of the elements of $\Lambda$ inside the $I_n$, while Equation \eqref{eq:energy} ensures that the $\lambda_n$ are reasonably equally spaced inside the $I_n$, by introducing the term $E_n$, which can be thought of as a repulsion term between the $\lambda_n$, which looks similar to the ones that appear in fermionic statistics.

The most general theorem of \cite{poltoratski2013problem} is then:
\begin{thm}
Let $\mu$ be a finite positive measure on $\mathbb{R}$. Let $a>0$. Then $T_\mu\geq a$ if and only if for all lower semi-continuous $W\in L^1(\mu)$ with limit $\infty$ at $\pm\infty$, and any $0<d<a$, there exists a $d$-uniform sequence $(\lambda_n)_{n\in\mathbb{N}}$ in the support of $\mu$ such that \begin{align}\sum_{n}\frac{\mathrm{log}\,W(\lambda_n)}{1+\lambda_n^2}<\infty.\end{align}
\end{thm}

In practise, this theorem is not always easy to apply, in large part due to the fact that the characterization it introduces needs to be verified for \textit{all} weights $W$. However, there are some situations in which $W$ can be removed from the statement. For example, in the case where the measure $\mu$ is discrete, we have
a much more explicit criterion. We first define the notion of $a$-regular and strongly $a$-regular sequence, which is a sequence whose overall density in the reals is approximately $a$:

\begin{defn}
A sequence $\Lambda=(\lambda_n)$ is $a$-regular if for all $\varepsilon>0$, if there is a sequence $I_n$ of intervals such that for all $n$,
\begin{align}
\left\lvert\frac{\#(\Lambda\cap I_n)}{\lvert I_n\rvert}-a\right\rvert\geq\varepsilon,
\end{align} then \begin{align}\sum_{n}\frac{\lvert I_n\rvert^2}{1+\mathrm{dist}(0,I_n)^2}<\infty.\end{align}
\end{defn}
\begin{defn}
A sequence $\Lambda=(\lambda_n)$ is strongly $a$-regular if \begin{align}\int \frac{n_\lambda(x)-ax}{1+x^2}dx<\infty,\end{align}where $n_\lambda(x)$ is the counting function of the support of $\Lambda$.
\end{defn}
The characterization of the exponential type is then:
\begin{prop}
Let \begin{align}
\mu:=\sum_nw(n)\delta_{\lambda_n}
\end{align}
be a finite measure, where $\Lambda=(\lambda_n)$ is a sequence of points such that for some $c>0$, $\lvert \lambda_i-\lambda_j\rvert>c$. Let
\begin{align}
D:=\mathrm{sup}\left\{D_\ast(\Lambda^\prime),\,\Lambda^\prime\subset\Lambda,\,\sum_{\lambda_n\in\Lambda^\prime}\frac{\mathrm{log}\,w(n)}{1+n^2}>-\infty\right\},
\end{align}
where
\begin{align}
D_\ast(\Lambda^\prime)=\mathrm{sup}\,\{a\geq0, \,\Lambda\,\text{has a strongly }a\text{-regular subsequence}\}.
\end{align}
Then, \begin{align}T_\mu=2\pi D.\end{align}
\end{prop}

It is this characterization that we used to show that below the Hawking--Page temperature, the depth parameter is equal to $\pi$. In cases where the sequence is not separated, i.e. where there is no minimal distance separating its points, we have an alternative result:

\begin{prop}
Let \begin{align}
\mu:=\sum_nw(n)\delta_{\lambda_n}
\end{align}
be a finite measure.
\begin{align}
D:=\mathrm{sup}\left\{d,\,\exists\,d-\text{uniform}\,\Lambda^\prime\subset\Lambda,\,\sum_{\lambda_n\in\Lambda^\prime}\frac{\mathrm{log}\,w(n)}{1+n^2}>-\infty\right\}.
\end{align}
Then, \begin{align}T_\mu\geq2\pi D,\end{align} and there is equality if \begin{align}\int dx\frac{\mathrm{log}(\lvert n_\Lambda\rvert+1)(x)}{1+x^2}<\infty.\end{align}
\label{prop:c7}
\end{prop}

There is also another useful result on the exponential type problem in the case of some polynomially growing measures called Frostman measures \cite{POLTORATSKI2019348}. These are positive measures $\mu$ for which there exist $C,\alpha>0$ such that for all intervals $I\subset\mathbb{R}$, \begin{align}\mu(I)<C\lvert I\rvert^\alpha.\end{align}In particular, Frostman measures include bounded continuous functions. We then have the following result:
\begin{thm}
Frostman measures have exponential type $0$ or $\infty$.
\end{thm}
In our language this means that if the spectral density of our theories at large $N$ is a bounded continuous function in the TFD state, then either there is no emergent radial direction at all, or there is stringy connectivity between both sides! 

The techniques used to prove the results described above are far from elementary, and make use of advanced theories in harmonic analysis, such as the theory of Toeplitz operators. They are also related to other deep areas of mathematical analysis, such as Beurling--Malliavin theory and the Gelfand--Levitan method, which we believe should also play a role in the definition of stringy geometries. The relation between these results and our physical setup shows that our depth parameter contains highly nontrivial information. We also hope that the new physical relevance of the exponential type found in this paper encourages further developments in the mathematical study of exponential type problems. 

\subsection{An explicit computation in a simple case}
\label{app:exp2}
Although the arguments of \cite{poltoratski2013problem} to establish the results above are very involved mathematically, in this appendix we offer an elementary argument to compute the exponential type for a simple choice of spectral function $\rho(\omega)$. Namely we assume

\be \label{vaespe}
\rho (\om) = 
\sum_{n=0}^\infty a_n \le[\delta(\omega-2n) -\delta(\omega+2n) \ri], 
\ee

which corresponds to the case \eqref{vae} of the vacuum state of a $d$-dimensional CFT on $S^{d-1}$ in the particular case $\omega_0=0$. In this case it can be readily seen that the depth parameter $\mathcal{T}$ is equal to $\pi$. More formally:

\begin{prop}
If the $a_n$ are nonzero and decay at most polynomially, the exponential type of the measure $\rho(\omega)$ is $\mathcal{T}=\pi$.
\end{prop}

\begin{proof}
We first establish that $\mathcal{T}\leq\pi$, and then establish that $\mathcal{T}\geq\pi$. First, $\mathcal{T}\leq\pi$. Indeed, suppose that there exists a function $f$ with real-valued Fourier transform such that the Fourier transform of $f(\omega)\rho(\omega)$ (or $\rho_0 (\om)$) vanishes on an interval of size larger than $\pi$. Introducing $c_n:=f_na_n$ and noticing that \begin{align}f(\omega)\rho(\omega)&=\sum_{n=0}^\infty c_n\delta(\omega-2n)-\sum_{n=0}^\infty c_n^\ast\delta(\omega+2n),\end{align} we obtain that $(f\ast\rho)(t)$ is periodic of period $\pi$. So if it vanishes on an interval of size larger than $\pi$, then it identically vanishes and $f=0$. Hence we have established that \begin{align}\mathcal{T}\leq\pi.\end{align} Now, we show that $\mathcal{T}\geq\pi$ by explicitly exhibiting, for all $\varepsilon>0$, test functions in $L^2(\rho)$ such that the Fourier transform of $f(\omega)\rho(\omega)$ vanishes on intervals of size $\pi-\varepsilon$. More precisely, suppose that $\mathcal{T}=\pi-\varepsilon$ for some $\varepsilon>0$. Then this means that there is no function $f(\omega)\in L^2(\rho)$ such that \begin{align}(f\ast\rho)(t)=i\,(\Pi\ast K)(t),\end{align} where $\Pi$ is the Dirac comb and $K$ is a compactly supported smooth function of support inside $(-\frac{\pi}{2},\frac{\pi}{2})$ but outside $(-\frac{\pi}{2}+\frac{\varepsilon}{2},\frac{\pi}{2}-\frac{\varepsilon}{2})$) (otherwise integration against such a function would vanish for functions supported on $(-\frac{\pi}{2}+\frac{\varepsilon}{2},\frac{\pi}{2}-\frac{\varepsilon}{2})$). So in Fourier space this means there are no functions $K$ and $f$ as defined before such that \begin{align}\sum a_{\pm n} f_{\pm n} \delta_{\pm 2n}=\sum K_{\pm n}\delta_{\pm 2n}.\end{align} But for any $K$ compactly supported in $t$, one can always find such an $f$: simply define $f_{\pm n}=K_{\pm n}/a_{\pm n}$. Since the $K_{\pm n}$ decay faster than polynomially (as $K$ is in Schwartz space), this implies \begin{align}\sum\frac{\lvert K_{\pm n}\rvert^2}{a_{\pm n}}<\infty,\end{align} so such an $f$ is in $L^2(\rho)$, which is a contradiction. Therefore, in that case, \begin{align}\mathcal{T}=\pi.\end{align}

\end{proof}

\newpage
\bibliographystyle{jhep}
\providecommand{\href}[2]{#2}\begingroup\raggedright\endgroup

\end{document}